%% file: AHJ-Perpetual-Futures-Pricing-R1.tex
\documentclass[12pt,fleqn,a4paper]{article}
\usepackage[longnamesfirst,round]{natbib} % bib style
\usepackage{import,standalone,mathpazo} % import other tex files (e.g. figures) without preambule
\usepackage{amsfonts,amsmath,mathtools,mleftright,mathrsfs,amssymb} % math stuffs
\usepackage[colorlinks=true,citecolor=blue,urlcolor=blue,linkcolor=blue,pagebackref]{hyperref}
\usepackage{graphicx,subfig,multirow} % graph stuffs
\usepackage{tikz,pgfplots} % 
\usepackage[]{mathtools} % 
\mathtoolsset{showonlyrefs}
\usepackage{enumitem,calrsfs}
\setitemize{noitemsep,topsep=0pt}
\usepackage{doi} % for doi link
\usepackage{fullpage,caption} % set the default 4 inches margin to be 1 inch
\pgfplotsset{
    compat=newest,
    scaled y ticks=false,
    scaled x ticks=false,
    yticklabel style={
        /pgf/number format/fixed,
        /pgf/number format/precision=2
	},
    xticklabel style={
        /pgf/number format/fixed,
        /pgf/number format/precision=2}}

\usepgfplotslibrary{dateplot}

  \definecolor{lgrey}{HTML}{B9C2CA}
  \definecolor{mgrey}{HTML}{8B96A2}
  \definecolor{dgrey}{HTML}{5D6974}
  \definecolor{aqua1}{HTML}{0f636d}
  \definecolor{aqua2}{HTML}{3ce0d5}
  
\pgfplotsset{
  yticklabel style={
        /pgf/number format/fixed,
        /pgf/number format/precision=2},
  scaled y ticks=false,
  compat=newest,
  mbaseplot/.style={
      legend style={draw=none,fill=none},
      x tick label style={font=\footnotesize},
      y tick label style={font=\footnotesize},
      x label style={font=\small}, 
      y label style={font=\small},
      legend style={draw=none,fill=none,font=\scriptsize},
      title style={font=\small},
      tick style={thick,black,on layer=axis foreground}, 
      major grid style={thin,lgrey,on layer=axis background}, 
	  minor grid style={thin,lgrey,on layer=axis background}, 
      axis x line=bottom,
      axis y line=left,
      axis line style={thick,black,on layer=main}},
  mlineplot/.style={
     mbaseplot,
     xmajorgrids=true,
     ymajorgrids=true,
     legend style={
     cells={anchor=west},
     draw=none}}}
\pgfdeclarelayer{background}
\pgfdeclarelayer{foreground}
\pgfdeclarelayer{fforeground}
\pgfsetlayers{background,main,foreground,fforeground}
\usepgfplotslibrary{fillbetween,groupplots}
\usepgfplotslibrary{colorbrewer}
\usetikzlibrary{matrix}
\usetikzlibrary{calc}
\makeatletter
\def\pgfplots@drawtickgridlines@INSTALLCLIP@onorientedsurf#1{}
\makeatother 
\makeatletter % https://tex.stackexchange.com/a/205193/156344
\newcommand*\short[1]{\expandafter\@gobbletwo\number\numexpr#1\relax}
\makeatother

\newcommand{\ind}[1]{\mathbf{1}_{\left\{#1\right\}}}

\renewcommand{\S}{{\mathcal S}}
\newcommand{\F}{{\mathcal F}}
\newcommand{\FF}{{\mathbb F}}
\newcommand{\N}{{\mathbb N}}
\newcommand{\R}{{\mathbb R}}

\newcommand{\Mloc}{\mathcal{M}_{\mathrm{loc}}}
\newcommand{\fignotes}[1]{
    \vspace{0\baselineskip}
    \begin{center}
    \parbox{0.9\textwidth}{\small
    \addtolength{\baselineskip}{0.15\baselineskip}
    #1}
    \end{center}}
    
\usepackage{ntheorem}
\theoremstyle{plain}
\newtheorem{theorem}{Theorem}
\newtheorem{corollary}{Corollary}

\newtheorem{lemmaA}{Lemma}[section]
\newtheorem{proposition}{Proposition}
\theorembodyfont{\upshape}

\newenvironment{proof}[1][\proofname]
{\par\normalfont\trivlist\item[\hskip\labelsep\textbf{#1}.]\ignorespaces}{\hfill$\blacksquare$\endtrivlist}
\newcommand{\proofname}{Proof}

\usepackage{setspace}
\title{Perpetual Futures Pricing\footnote{First version: October 30, 2023. This paper merges and extends two independent papers with the same title written by the first two authors in July 2023 and by Urban Jermann in September 2023. Julien Hugonnier collectively thanks the Spring 2023 class of the FIN404 Derivatives course in the Master of Financial Engineering at EPFL for excellent research assistance on the topic of this paper.}}
\author{
Damien Ackerer\footnote{EPFL and Swissblock Technologies AG. Email:  dackerer@swissblock.net}
\and
Julien Hugonnier\footnote{EPFL and CEPR. Email: julien.hugonnier@epfl.ch}
\and
Urban Jermann\footnote{Wharton and NBER. Email: jermann@wharton.upenn.edu}
}
\date{September 3, 2024\vspace{-0.5cm}}%\today\vspace{-0.5cm}}

\usepgfplotslibrary{external}
\tikzsetfigurename{AJH-Figure_}
\expandafter\gdef\csname c@tikzext@no@\pgfkeysvalueof{/tikz/external/figure name}\endcsname{1}%
\begin{document}
% ---------------------------------------------------
\maketitle

\begin{abstract}
\noindent
Perpetual futures are contracts without expiration date in which the anchoring of the futures price to the spot price is ensured by periodic funding payments from long to short. We derive explicit expressions for the no-arbitrage price of various perpetual contracts, including linear, inverse, and quantos futures in both discrete and continuous-time. In particular, we show that the futures price is given by the risk-neutral expectation of the spot sampled at a random time that reflects the intensity of the price anchoring. Furthermore, we identify funding specifications that guarantee the coincidence of futures and spot prices, and show that for such specifications perpetual futures contracts can be replicated by dynamic trading in primitive securities.\\[-0.4\baselineskip]

\noindent \textbf{Keywords.} Cryptocurrencies, Derivatives, Futures contracts.\\[-0.4\baselineskip]

\noindent \textbf{JEL Classification.} E12, G13.

\end{abstract}

\clearpage
% ---------------------------------------------------
\section{Introduction}
% ---------------------------------------------------
\shortcites{he2022fundamentals,paradigm2021power}

% long intro on perps
Perpetual futures contracts are financial derivatives that offer the same characteristics as traditional futures
contracts but without an expiration date. Like traditional futures, perpetual futures allow traders to speculate at no
cost on the price movements of an underlying asset without actually holding it. The fact that the contract does not
have a fixed maturity date presents two related advantages. First, it implies that traders can take positions for the
duration of their choice without having to rollover from maturing contracts to newly minted contracts and thus
tremendously simplifies the investment process. Second, and related, the fact that a single contract is traded on each
underlying fosters a higher liquidity which in turn facilitates price discovery. At the time of writing, perpetual
futures are particularly popular in cryptocurrency markets (see Figure \ref{figure:BTV perpetual New} for some data on
BTC/USD futures) but we expect that they will soon find traction in other asset classes such as commodities.

In a well-functioning perpetual futures market, the contract's price should closely track the spot price of the
underlying. However, the demand and supply dynamics at play in the market imply that there can be temporary deviations
between the futures price and the spot. These deviations may lead to a \emph{premium} or positive basis when the
futures prices price is higher than the spot, or to a \emph{discount} when the futures price is lower than the spot. In
traditional futures, the existence of a finite maturity forces the futures price to converge to the spot at the
expiry date, and this terminal contraint effectively limits the size of the basis. In a perpetual futures contract
without maturity the anchoring of the futures price to the spot is instead achieved through periodic \emph{funding
payments} from the long to the short. These funding payments are computed periodically (e.g., every 8h on most
trading platforms) as the sum of a \emph{premium term} that depends on the spread between the futures
price and the spot, and of an \emph{interest term} that reflects the interest rate differential between the base and quote currency. See \citet{Medium-1,Medium-2}, \citet{Naka}, \citet{White-Cartoon}, and \citet{Medium-3} for practioner-oriented discussions of perpetual futures and their uses.

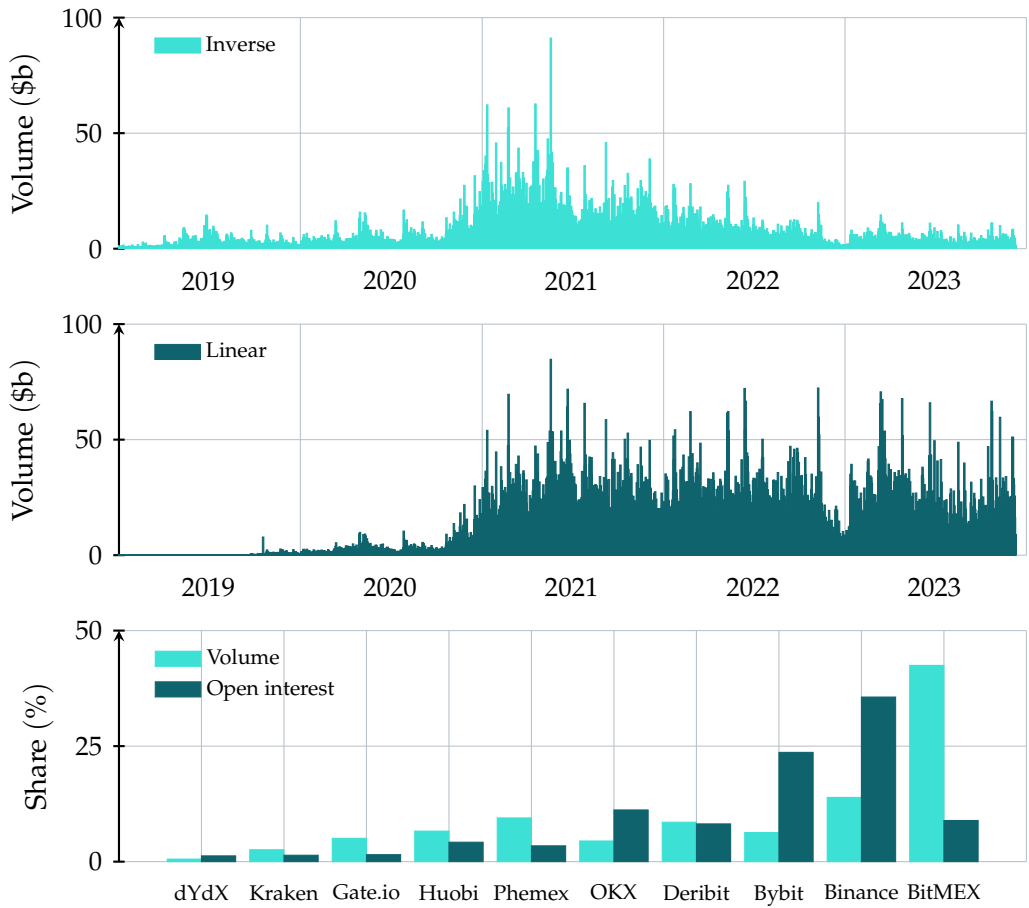
\begin{figure}[t!]
\centering
\input{AHJ-Figure1.tex}
\caption{Perpetual Bitcoin futures.}\label{figure:BTV perpetual New}
\fignotes{The top two panels show the evolution of the aggregated daily trading volume in perpetual linear and inverse bitcoin futures from January 1, 2019 to December 10, 2023. The bottom panel shows the average repartition of trading volume and open interest across the ten largest trading platforms over the same period.}
\end{figure}

The terms of a standard, or \emph{linear}, perpetual futures contract include the underlying asset (e.g., BTC/USD) representing the value of one unit of the base currency (BTC) in units of the quote currency (USD), a contract size expressed in units of the base asset (1BTC), and a margin and settlement currency (USD) in which profits and losses are realized. Cryptocurrency trading platforms have introduced multiple variations of the linear contracts. The most important such variation is the \emph{inverse} contract where the base currency itself (BTC) is used as the margin and settlement asset and the contract size is expressed in units of the quote currency (10'000USD). This innovative product allows to speculate on the exchange rate between a crypto and a fiat currency without the need to actually hold units of that fiat currency and was thus widely adapted by trading platforms that typically cannot accept deposits in fiat currencies since they do not qualify as banks under the existing regulation.  
Another important variation is the perpetual \emph{quanto} futures that uses a third currency (ETH), different from the quote and base currencies, for margining and settlement. Perpetual futures contracts where the target spot price is a function of the spot price (rather than the spot price itself) have been proposed by \citet{paradigm2021everlasting} under the somehow misleading name of \emph{everlasting options} but, outside of power contracts \citep{paradigm2021power,Perennial} and their applications \citep{medium2022sqeeth,Clark2,Clark} they have found limited traction so far. 

% paper contribution
In this paper, we derive the prices of different perpetual futures contracts under the absence of arbitrage. We study discrete-time and continuous-time formulations. We identify funding specifications for the linear and inverse contracts that guarantee that the perpetual future price coincides with the spot so that the basis is constantly equal to zero. In both cases, the required interest term is an easily implementable function of the interest rates in the two currencies. With the assumption of constant funding parameters and interest rates, we derive explicit model-free expressions for the linear and inverse futures prices. We show that, in general, the perpetual future price is the discounted expected value of the future underlying asset's price at a random time that reflects the funding specification. In the continuous-time case, we provide a general expression for the quanto futures price and use it to obtain a closed-form solution for the required convexity correction in a Black-Scholes setting. Finally, we derive general pricing formulas for everlasting options that we illustrate with closed form solutions for calls and puts in a Black-Scholes setting.

The first case of futures contracts without maturity can be found at the Chinese Gold and Silver Exchange of Hong Kong
who developed an \emph{undated futures} market. However, these contracts did not allow the futures price to fluctuate
freely and would instead be settled everyday against the spot price with an interest payment analog to the funding
payment. In essence, these undated contracts were automatically rolled over one-day futures contracts, see
\citet{gehr1988undated} for a discussion. Perpetual futures contracts were formally introduced by
\citet{shiller1993measuring} who proposed the creation of perpetual claims on economic indicators (such as real estate
prices and/or corporate profits) where the funding rate would depends on observable cash-flows such as rental rates.
While the focus of \citet{shiller1993measuring} was on real economic indicators rather than cryptocurrencies, his paper
laid the groundwork for the subsequent development of perpetual futures. BitMEX is credited with pioneering and
popularizing perpetual contracts for cryptocurrencies. In particular, that platform introduced the inverse contracts in
2016, see \citet{BitMEX-Intro}. At the time of writing, perpetual futures are listed on dozens of exchanges and
constitute by far the dominant derivatives instrument in that space. For example, of all the listed futures contracts
on Bitcoin traded during the first half of 2023, $75\%$ of the 27B USD daily average volume and $94\%$ of the 8B USD
daily average open interests can be attributed to perpetual futures.

% literature review
The academic literature on the pricing of perpetual futures contracts is very scarce. \citet{angeris2023primer} study a
related but different problem in a no-arbitrage framework. They derive the funding rate value such that the perpetual
future price is explicitly given by a function of the spot price and the unobservable parameters of the prices process.
This type of contract is different from the perpetual future contracts actually traded in cryptocurrency markets, and
is not implemented by any exchange at the time of writing. \citet{he2022fundamentals-original} compute the perpetual
futures price under the assumption of a constant funding rate and perform an empirical study of the deviations between
their theoretical price and market observations \citep[see also][]{java}. The perpetual futures price they derived in
the initial version of their paper was incorrect as it relied on a specification of cash flows that is incompatible
with the assumption that entering the contract is costless (see Appendix \ref{sec:An erroneous cash flow specification}
for a review). After being made aware of this error, \citeauthor{he2022fundamentals-original} amended their
specification and the most recent version of their work \citep{he2022fundamentals} features a correct perpetual
futures pricing formula that is a special case of the result first derived in this paper. To the best of our knowledge
there are currently no other theoretical studies of perpetual inverse and quanto futures contracts. Despite the
scarcity of the theoretical literature there are several empirical studies of perpetual futures focusing on price
discovery \citep{alexander2020bitmex}, crypto carry trades \citep{schmeling2022crypto,Nicholas023}, and the microstructure
of perpetual futures markets \citep{deblasis2022arbitrage}.

% paper structure
The paper is split in two parts. In the first we consider a discrete-time formulation of a market with two currencies and derive explicit expressions for linear and inverse perpetual futures prices. In the second part we move to a continuous-time formulation in which we derive expressions for linear, inverse, and quantos futures prices that we compare with the recent results of \citet{he2022fundamentals} and \citet{angeris2023primer}. Finally, we briefly consider everlasting options and illustrate their pricing within a standard lognormal setting. The proofs of all results are provided in the appendix.

\part{Discrete-time}

\section{The underlying model} % (fold)
\label{sec:the_discrete_time_model}

Time is discrete and indexed by $t\in\{0,1,\dots\}$. Uncertainty in the economy is captured by a probability space $(\Omega,\F,P)$ that we equip with a filtration $\FF=\{\F_t\}_{t\geq 0}$. Unless specified otherwise all stochastic processes to appear in what follows are implicitly assumed to be adapted to $\FF$.

There are two currencies indexed by $i\in\{a,b\}$, for example the US Dollar and Bitcoin. We denote by $x_t>0$ the $b/a$ exchange rate at date $t$, that is the price in $a$ of 1 unit of $b$. Investors can freely exchange currencies at this rate and are allowed to invest in two locally risk free bonds: one denominated in units of $a$ and the other in units of $b$. The price of these assets satisfy
\begin{align}
  B_{it+1}=\left(1+r_{it}\right) B_{it},\qquad B_{i0}=1, \qquad i\in\{a,b\},
\end{align}
where $r_{it}>-1$ is the $\F_{t}-$measurable return on the $i-$denominated risk free asset over the period from date $t$ to date $t+1$.

To ensure the absence of arbitrages between these two primitive assets, we assume that there exists a probability $Q_a$ that is equivalent to $P$ when restricted to $\F_t$ for any fixed $t\in\N$ and such that the price of the $b-$riskless asset expressed in units of $a$ and discounted at the $a-$risk free rate is a martingale under $Q_a$:
\begin{align}\label{eq:NA Qa DT}
E^{Q_a}_t\left[ \frac{B_{bs}}{B_{as}}x_{s}   \right ]=\frac{B_{bt}}{B_{at}}x_t,    \qquad  \forall t\leq s.
\end{align}
Note that the choice of $a$ as the reference currency is without loss. Indeed, using the strictly positive $Q_a-$martingale $\frac{B_{bt}}{B_{at}}\frac{x_t}{x_0}$ as a density process shows that under the above no-arbitrage assumption there exists a probability $Q_b$ that is equivalent to $P$ when restricted to $\F_t$ for any fixed $t$ and such that
\begin{align}\label{eq:NA Qb DT}
E^{Q_b}_t\left[ \frac{B_{as}}{B_{bs}}x^*_{s} \right ]=\frac{B_{at}}{B_{bt}}x^*_t,    \qquad  \forall t\leq s,
\end{align}
where $x^*_t\equiv 1/x_t$ denote the $a/b$ exchange rate. We will have the occasion to use both of these currency-specific pricing measures in what follows.

% section the_discrete_time_model (end)

\section{Perpetual futures pricing} % (fold)
\label{sec:linear_futures_pricing-DT}

A \emph{perpetual linear} (or direct) futures contract provides exposure to one unit of currency $b$ from the point of view of an investor whose unit of account is $a$. Accordingly, the perpetual futures price $f_t$ is quoted in units of currency $a$ and all margining operations required by the contract are carried out in that currency. In principle, $a$ and $b$ can be either fiat or crypto currencies. In practice, however, the reference currency $a$ is most often a pure crypto currency or a stablecoin equivalent of a fiat currency, such as USDT, because exchanges typically do not operate in fiat currencies.

Entering a contract at date $t$ is costless and, as in a classical futures contract, the long receives at date $t+1$ the one period variation 
\begin{align}
  f_{t+1}-f_t
\end{align}
in the futures price. If the futures contract had a finite maturity date, say $T$, then this periodic cash flow and the condition that the futures price should equal the spot at maturity are sufficient to uniquely pin down the futures price as the conditional expectation
\begin{align}\label{eq:finite maturity futures price}
  f^T_t:=E^{Q_a}_t\left[x_{T}\right]
\end{align}
of the terminal spot price under the $a-$risk neutral probability $Q_a$. Fixed maturity contracts on commodities and fiat/fiat currency pairs are widely traded on different platforms, but at the time of writing the only existing such contracts on crypto/fiat currency pairs trade on the Chicago Mercantile Exchange \citep{CME}.

Without a fixed maturity one needs an alternative mechanism to keep the futures price anchored to the underlying spot price. This is achieved by the introduction of a periodic \emph{funding} payment. The exact specification of this funding payment varies across exchanges but generally consists in two predictable components: A \emph{premium} part and an \emph{interest} part. At date $t+1$ the premium part is
\begin{align}\label{eq:premium part}
 \kappa_{t} \left(f_t - x_t\right)
\end{align}
where the rate $\kappa_{t}>0$ controls the strength of the anchoring of the futures price to the spot price. This premium component ensures that the futures prices remains close to the spot price by introducing an automatic correction mechanism that is similar to mean reversion. Intuitively, if the futures price is above the spot price at date $t$ then long investors will have to pay the amount $\kappa_t\left(f_t-x_t\right)>0$ at date $t+1$. This makes the short side more attractive and thus generates an excess demand for short positions which ultimately lowers the futures price. Likewise, if the futures price is lower than the spot price  at date $t$ then holders of long positions know they will receive the amount $-\kappa_t\left(f_t-x_t\right)>0$
at date $t+1$ which makes the long side more attractive and induces the futures price to increase. On the other hand, the interest part of the funding payment is given by 
\begin{align}\label{eq:interest part}
  \iota_t x_t
\end{align}
where the factor $\iota_t$ is set by the exchange to reflect the possibly time varying interest rate differential between the two currencies.

In accordance with these definitions, the $a-$denominated cash flow at date $t+1$ from holding a long position over the period from $t$ to $t+1$ is
\begin{align}\label{eq:CF one period}
\left(f_{t+1}-f_t\right)-\kappa_t \left(f_t - x_t\right)-\iota_t x_t
\end{align}
and, since entering a futures position is costsless, the absence of arbitrage between the futures, spot, and financial markets requires that 
\begin{align}
  E^{Q_a}_t\left[\left(f_{t+1}-f_t\right)-\kappa_t \left(f_t - x_t\right)-\iota_t x_t\right]=0, \qquad \forall t\geq 0.
\end{align} 
Rearranging this equality gives
\begin{align}
  f_t= \frac{1}{1+\kappa_t} E^{Q_a}_{t}\left[  f_{t+1}\right]+\left(\frac{\kappa_t-\iota_t}{1+\kappa_t}\right)  x_t
\end{align}
and iterating this relation forward reveals that
\begin{align}\label{eq:pricing restriction in DT}
  f_t= E^{Q_a}_{t}\left[ \left(\prod_{\tau=t}^{T-1} \frac{1}{1+\kappa_\tau}\right)  f_{T} +\sum_{\sigma=t}^{T-1} \left(\prod_{\tau=t}^{\sigma} \frac{1}{1+\kappa_\tau}\right) \left(\kappa_\sigma-\iota_\sigma\right) x_{\sigma}\right]
\end{align}
for all $0<t\leq T-1$. The solution to this recursive integral equation cannot be uniquely determined without imposing further constraints. In particular, it is easily seen that if $f_t$ is a solution to \eqref{eq:pricing restriction in DT} then 
\begin{align}
	 \hat f_t(\beta)\equiv f_t + \beta \prod_{\tau=1}^{t-1}\left(1+\kappa_\tau\right)
\end{align}
gives another solution for any $\beta\geq 0$, but this solution diverges and, thus, cannot constitute a viable futures price process. To rule out such solutions we impose the \emph{no-bubble} condition:
\begin{align}\label{eq:TVC in DT}
\lim_{T\to\infty}E^{Q_a}_{t}\left[ \left(\prod_{\tau=t}^{T-1} \frac{1}{1+\kappa_\tau}\right)  f_{T}\right]=0, \qquad  \forall t\geq 0. 
\end{align}
Standard arguments then lead to the following results:

\begin{theorem}\label{theorem:perpetual futures price DT-1}
  Assume that 
  \begin{align}
  E^{Q_a}\left[ \sum_{\sigma=0}^{\infty} \left(\prod_{\tau=0}^{\sigma} \frac{1}{1+\kappa_\tau}\right) \left|\kappa_\sigma-\iota_\sigma\right| x_{\sigma}\right]<\infty.\label{eq:integrability DT}
  \end{align}
  Then the process
  \begin{align}\label{eq:price as conditional expectation}
    f_t=E^{Q_a}_{t}\left[ \sum_{\sigma=t}^{\infty} \left(\prod_{\tau=t}^{\sigma} \frac{1}{1+\kappa_\tau}\right)  \left(\kappa_\sigma-\iota_\sigma\right)x_{\sigma}\right]
  \end{align}
   is the unique solution to \eqref{eq:pricing restriction in DT} that satisfies \eqref{eq:TVC in DT}. If, in addition, $-1<\iota_t<\kappa_t$ then this solution can be represented as
	  \begin{align}
	  	f_t= E^{Q_a}_t\left[\left(\prod_{\tau=t}^{\theta_t-1} \frac{1}{1+\iota_\tau}\right) x_{\theta_t}\right]
	  \end{align}
	  where $\theta_t$ is a random time that is distributed according to
	  \begin{align}
		  Q_{a}\left(\theta_t=\sigma\middle| \F\right)
		  =\ind{t\leq \sigma}\frac{\kappa_\sigma-\iota_\sigma}{1+\iota_\sigma}\left(\prod_{\tau=t}^{\sigma} \frac{1+\iota_\tau}{1+\kappa_\tau}\right).
	  \end{align}
    In particular, if the premium rate $\kappa$ is constant and the interest factor $\iota=0$ then the perpetual futures price is given by 
	  \begin{align}
	  	f_t=E^{Q_a}_t\left[\sum_{n=0}^\infty \kappa \left(1+\kappa\right)^{-(n+1)} x_{t+n}  \right] =E^{Q_a}_t\left[x_{t+\theta}\right]
	  \end{align}
	  where $\theta$ is a geometrically distributed random variable with mean $1/\kappa$.
\end{theorem}

Like most of our results, Theorem \ref{theorem:perpetual futures price DT-1} requires two related conditions: the no-bubble condition \eqref{eq:TVC in DT} to single out a unique solution to \eqref{eq:pricing restriction in DT} and the integrability condition \eqref{eq:integrability DT} to ensure that this solution is well-defined and not simply $\pm\infty$ at all times. Neither of these conditions can be relaxed if one is to derive a meaningful futures price. Indeed, \eqref{eq:integrability DT} is clearly necessary for a well-defined solution to exist and, under this conditions, the bubble-free solution in \eqref{eq:price as conditional expectation} is the only one that does not diverge. 

In the standard case of a finite maturity contract the futures price is simply the expectation of the terminal spot price under $Q_a$ as in \eqref{eq:finite maturity futures price}. The second part of Theorem \ref{theorem:perpetual futures price DT-1} reveals that a similar representation holds for the perpetual futures price albeit with a random maturity whose distribution reflects the funding parameters of the perpetual contract. In particular, if the premium rate $\kappa$ is constant and $\iota=0$ then the perpetual futures price can be obtained in closed form within any model that admits an explicit expression for the finite maturity futures price $f_t^{t+n}=E^{Q_a}_t[x_{t+n}]$. 

If the funding coefficients and the interest rates are constant then the expectation on the right handside of \eqref{eq:price as conditional expectation} can be computed in closed form without the need to specify the stochastic dynamics of the underlying exchange rate:

\begin{proposition}\label{proposition:perpetual futures price constant parameters DT}
	If $\iota<\kappa$ and $(r_a,r_b)$ are constants such that
    \begin{align}\label{eq:integrability DT-constant}
    	\frac{1}{1+\kappa}\left(\frac{1+r_a}{1+r_b}\right)<1
    \end{align}
    then the perpetual futures price
    \begin{align}
    f_t
	=\frac{\kappa-\iota}{1+\kappa}\, E^{Q_a}_{t}\left[ \sum_{\sigma=t}^{\infty} \left(\frac{1}{1+\kappa}\right)^{\sigma-t}  x_{\sigma}\right]
  	=\frac{\left(\kappa-\iota\right)\left(1+r_b\right)}{r_b-r_a+\kappa\left(1+r_b\right)} x_t\label{eq:price DT-constant}
    \end{align}
is increasing in $r_a$ as well as decreasing in $r_b$ and $\iota$, and converges monotonically to the spot price as the premium rate $\kappa\to\infty$.
\end{proposition}

Condition \eqref{eq:integrability DT-constant} is the counterpart of the integrability condition \eqref{eq:integrability DT} when the interest rate and funding parameters are constant. Indeed, in this case the no-arbitrage restriction \eqref{eq:NA Qa DT} implies that
\begin{align}
  E^{Q_a}\left[  \left(\prod_{\tau=0}^{\sigma} \frac{1}{1+\kappa}\right) \left|\kappa-\iota\right| x_{\sigma}\right]
  % &=\frac{\left|\kappa-\iota\right|E^{Q_a}\left[  x_{\sigma}\right]}{(1+\kappa)^\sigma} \\
  &=\left[\frac{1}{1+\kappa}\left(\frac{1+r_a}{1+r_b}\right)\right]^\sigma \left|\kappa-\iota\right| x_0
\end{align}
and it follows that the sum
\begin{align}
  \sum_{\sigma=0}^T E^{Q_a}\left[  \left(\prod_{\tau=0}^{\sigma} \frac{1}{1+\kappa}\right) \left|\kappa-\iota\right| x_{\sigma}\right]=\left|\kappa-\iota\right| x_0\sum_{\sigma=0}^T\left[\frac{1}{1+\kappa}\left(\frac{1+r_a}{1+r_b}\right)\right]^\sigma 
\end{align}
converges if and only if condition \eqref{eq:integrability DT-constant} is satisfied. On the other hand, the condition that $\kappa-\iota>0$ is required to ensure that the futures price process is positive at all times. To see this, simply note that \eqref{eq:price DT-constant} can be written as
\begin{align}
  f_t=\frac{\left(\kappa-\iota\right)x_t}{\left(1+\kappa\right)(1-\psi)}
\end{align}
where the constant $\psi<1$ is defined by the left side of \eqref{eq:integrability DT-constant}.

\begin{figure}[t!]
	\centering
	\input{AHJ-Figure2.tex}
	\label{fig:futures price}
\end{figure}	

To illustrate the results of Proposition \ref{proposition:perpetual futures price constant parameters DT} we plot in Figure \ref{fig:futures price} the ratio $f_t/x_t$ as a function of the anchoring intensity $\kappa$ for a contract with $\iota\equiv 0$ and different interest rate configurations. Without an interest factor \eqref{eq:price DT-constant} simplifies to
\begin{align}
  f_t/x_t=\frac{\kappa\left(1+r_b\right)}{\kappa\left(1+r_b\right)-\delta}
\end{align}
where $\delta:=r_a-r_b$ denotes the interest rate spread between currencies $a$ and $b$. With this formula at hand, Figure \ref{fig:futures price} shows that the futures price is above the spot if $\delta>0$,  equal to the spot if $\delta=0$, and below the spot when $\delta>0$. The fact that this comparison is driven by the interest rate spread is intuitive. For example, if $\delta<0$ a situation of \emph{contango} arises from the fact that a higher interest on currency $b$ makes it more attractive to hold the currency directly rather than the futures contract which, in turn, tends to lower the futures price. Likewise, if $\delta>0$ a situation of \emph{backwardation} arises from the fact that the lower interest rate on $b$ makes that currency less attractive than the futures contract. Figure \ref{fig:futures price} further shows that the magnitude of the deviation from the spot price is increasing in the magnitude of the spread $|\delta|$ and decreasing in the anchoring intensity $\kappa$ with $\lim_{\kappa\to\infty}|f_t-x_t|=0$.

Next, we show that the interest factor can be chosen by the exchange in such a way that the perpetual futures price and the spot price coincide for any sufficiently large premium rate $\kappa_t$. This result is important from the market design point of view because it delivers a perfect anchoring of the futures price to the spot price through a simple specification of the sole interest factor. It is also important from the financial engineering point of view because if the contract is such that $f_t=x_t$ then the perpetual futures contract can be dynamically replicated by trading in the two primitives assets despite any potential market incompleteness. 

\begin{corollary}\label{corollary:perpetual futures price equalize DT}
If the interest factor
\begin{align}
	\iota_t=\frac{r_{at}-r_{bt}}{1+r_{bt}}<\kappa_t
\end{align}
then the futures price $f_t=x_t$ at all times. In this case, the one period cash flow of a long position can be replicated by borrowing
\begin{align}
	 m_t x_t=\frac{x_t}{1+r_{bt}}
\end{align}
units of $a$ at rate $r_{at}$ and investing $m_t$ units of $b$ at rate $r_{bt}$.
\end{corollary}

% In the standard case of a finite maturity contract the futures price is simply given by the risk-neutral expectation of the terminal spot price under $Q_a$ as in \eqref{eq:finite maturity futures price}. Our last result in this section shows that a similar representation holds for the perpetual futures price albeit with a random maturity date whose distribution reflects the funding parameters of the perpetual contract.

% \begin{theorem}\label{theorem:perpetual futures price DT-2}
% Assume that \eqref{eq:integrability DT} holds and that $-1<\iota_t<\kappa_t$. Then the perpetual futures price satisfies
% 	  \begin{align}
% 	  	f_t= E^{Q_a}_t\left[\left(\prod_{\tau=t}^{\theta_t-1} \frac{1}{1+\iota_\tau}\right) x_{\theta_t}\right]
% 	  \end{align}
% 	  where $\theta_t\geq t$ is a random time that is defined on an extension $\widetilde \Omega$ of the probability space and distributed according to
% 	  \begin{align}
% 		  Q_{a}\left(\theta_t=\sigma\middle| \F\right)
% 		  =\ind{t\leq \sigma}\frac{\kappa_\sigma-\iota_\sigma}{1+\iota_\sigma}\left(\prod_{\tau=t}^{\sigma} \frac{1+\iota_\tau}{1+\kappa_\tau}\right).
% 	  \end{align}
%     In particular, if the premium rate $\kappa$ and the interest factor $\iota<\kappa$ are constants then the perpetual futures price is simply given by
% 	  \begin{align}
% 	  	f_t=E^{Q_a}_t\left[\sum_{n=0}^\infty \kappa \left(1+\kappa\right)^{-(n+1)} x_{t+n}  \right] =E^{Q_a}_t\left[x_{t+\eta}\right]
% 	  \end{align}
% 	  where $\eta:\widetilde \Omega\to\N$ is a geometrically distributed random time with mean $1/\kappa$.
% \end{theorem}

We close this section with a remark regarding the robustness of our specification of funding payments.
In line with the interpretation that these payments are like interest payments, our specification in \eqref{eq:premium part} and \eqref{eq:interest part} is fully predictable: The amount that will be paid or received at date $t+1$ is known to all market participants at date $t$. This, however, is not the only possible specification. In particular, some exchanges specify the funding payment as the product of a predictable funding rate and a \emph{mark value} that is approximately equal to the spot price at the end of the funding period.  With such a specification the one period cash flow in \eqref{eq:CF one period} must be replaced by
  \begin{align}\label{eq:CF one period mark-value}
  \left(f_{t+1}-f_t\right)-  x_{t+1} \hat\iota_t-x_{t+1} \hat\kappa_t \left(\frac{f_t - x_t}{x_t}\right)
  \end{align}
  for some \emph{adapted} $(\hat\iota_t,\hat\kappa_t)$ and, since
  \begin{align}
  	E^{Q_a}_t\left[x_{t+1}\hat\iota_t+x_{t+1}\hat\kappa_t \left(\frac{f_t - x_t}{x_t}\right)\right]=\frac{1+r_{at}}{1+r_{bt}}\left(\hat\iota_t x_t+\hat\kappa_t \left(f_t-x_t\right)\right),
  \end{align}
  we have that the associated perpetual futures price can be computed as in Theorem \ref{theorem:perpetual futures price DT-1} but with the adjusted rates
  \begin{align}
  	\left(\iota_t,\kappa_t\right)=\frac{1+r_{at}}{1+r_{bt}}\left(\hat\iota_t,\hat\kappa_t\right).
  \end{align} 
  Subject to this parameter adjustment, all our other results also apply to this alternative specification of funding payments.

% section linear_futures_pricing (end)

\section{Inverse futures pricing} % (fold)
\label{sec:inverse_futures_pricing-DT}

The inverse perpetual futures contract offers an exposure to the $b/a$ exchange rate and is quoted in units of currency $a$ but, unlike the linear contract, it is margined and funded in currency $b$. In this alternative form of contract, $a$ and $b$ can be in principle be fiat or crypto currencies. However, in practice $b$ is most often a crypto currency, such as BTC or ETH, because exchanges typically do not operate in fiat currencies, but $a$ can be of either type as its exchange rate only serves to compute funding payments denominated in units of $b$. As a result, inverse contracts are particularly well suited to crypto currency investors. Indeed, the fact that the inverse contract operates entirely in the target currency implies that it can be run entirely on chain without ever having to own or transfer any units of fiat currency. 

The contract size is expressed in units of $a$ and fixed to one.
As a result, the $b-$denominated cash flow at date $t+1$ from holding a long position in the inverse perpetual futures over the period from date $t$ to date $t+1$ is
\begin{align}
\left(\frac1{f_{It+1}}-\frac1{f_{It}}\right)-\kappa_{It} \left(\frac1{f_{It}} - x^*_t\right)-\iota_{It} x^*_t
\end{align}
where $f_{It}$ denotes the inverse perpetual futures price quoted in units of $a$, $x^*_t=1/x_t$ denotes the price of one unit of $a$ in units of $b$, and $(\iota_{It},\kappa_{It})$ are contract-specific adapted funding parameters set by the exchange. Since entering an inverse futures position is costless, the absence of arbitrage requires that
\begin{align}
E^{Q_b}_t\left[\left(\frac1{f_{It+1}}-\frac1{f_{It}}\right)-\kappa_{It} \left(\frac1{f_{It}} - x^*_t\right)-\iota_{It} x^*_t\right]=0, \qquad \forall t\geq 0,
\end{align}
where $Q_b$ is the pricing measure for $b-$denominated cash flows. Rearranging this identity we find that
\begin{align}
	\frac{1}{f_{It}}=\frac{1}{1+\kappa_{It}}E^{Q_b}_t\left[\frac{1}{f_{It+1}}\right]+\left(\frac{\kappa_{It}-\iota_{It}}{1+\kappa_{It}}\right) x^*_t
\end{align}
and iterating this relation forward reveals that
\begin{align}\label{eq:pricing restriction inverse in DT}
	\frac{1}{f_{It}}= E^{Q_b}_t\left[\left(\prod_{\tau=t}^{T-1}\frac{1}{1+\kappa_{I\tau}}\right)\frac{1}{f_{IT}} +\sum_{\sigma=t}^{T-1}\left(\prod_{\tau=t}^{\sigma}\frac{1}{1+\kappa_{I\tau}}\right)\left(\kappa_{I\sigma}-\iota_{I\sigma}\right) x^*_\sigma  \right].
\end{align}
for all $t+1\leq T$. As in the linear case, we single out a natural solution to this recursive equation by imposing the no-bubble condition
\begin{align}\label{eq:TVC inverse in DT}
\lim_{T\to\infty}E^{Q_b}_{t}\left[ \left(\prod_{\tau=t}^{T-1} \frac{1}{1+\kappa_{I\tau}}\right)  \frac{1}{f_{IT}}\right]=0, \qquad  \forall t\geq 0. 
\end{align}
Arguments similar to those of Section \ref{sec:linear_futures_pricing-DT} then deliver the following counterparts of Theorem \ref{theorem:perpetual futures price DT-1}, Proposition \ref{proposition:perpetual futures price constant parameters DT},  and Corollary \ref{corollary:perpetual futures price equalize DT} for the inverse contract.

\begin{theorem}\label{theorem:inverse future -DT1}
  Assume that 
  \begin{align}
    E^{Q_b}\left[ \sum_{\sigma=0}^{\infty} \left(\prod_{\tau=0}^{\sigma} \frac{1}{1+\kappa_{I\tau}}\right) \left|\kappa_{I\sigma}-\iota_{I\sigma}\right| x^*_{\sigma}\right]<\infty.\label{eq:integrability inverse DT}
  \end{align}
  Then the process
  \begin{align}
    \frac{1}{f_{It}}=E^{Q_b}_{t}\left[ \sum_{\sigma=t}^{\infty} \left(\prod_{\tau=t}^{\sigma} \frac{1}{1+\kappa_{I\tau}}\right)  \left(\kappa_{I\sigma}-\iota_{I\sigma}\right)x^*_{\sigma}\right]
  \end{align}
 is the unique solution to \eqref{eq:pricing restriction inverse in DT} that satisfies \eqref{eq:TVC inverse in DT}. If, in addition, $-1<\iota_t<\kappa_t$ then this solution can be represented as
	  \begin{align}
	  	  	\frac1{f_{It}}= E^{Q_b}_t\left[\left(\prod_{\tau=t}^{\theta_{It}-1} \frac{1}{1+\iota_{I\tau}}\right) x^*_{\theta_{It}}\right]
	  \end{align}
where $\theta_{It}$ is a random time that is distributed according to
	  	  \begin{align}
	  		  Q_{b}\left(\theta_{It}=\sigma\middle|\F\right)
	  		  =\ind{t\leq \sigma}\frac{\kappa_{I\sigma}-\iota_{I\sigma}}{1+\iota_{I\sigma}}\left(\prod_{\tau=t}^{\sigma} \frac{1+\iota_{I\tau}}{1+\kappa_{I\tau}}\right).
	  	  \end{align}
In particular, is constant and the interest factor is zero then the perpetual inverse futures price is simply given by 
\begin{align}
	1/f_{It}=E^{Q_b}_t\left[\sum_{n=0}^{\infty} \kappa_I\left(1+\kappa_I\right)^{-(n+1)}x^*_{t+n}\right]=E^{Q_b}_t[x^*_{t+\eta_I}]
\end{align}
where $\eta_I$ is geometrically distributed random variable with mean $1/\kappa_{I}$.		  
\end{theorem}

\begin{proposition}\label{proposition:perpetual inverse futures price constant parameters DT}
	If $\iota_I<\kappa_{I}$ and $(r_a,r_b)$ are constants such that
	    \begin{align}\label{eq:integrability DT-constant inverse}
	    	\frac{1}{1+\kappa_{I}}\left(\frac{1+r_b}{1+r_a}\right)<1
	    \end{align}
	    then the inverse futures price
	    \begin{align}
	    f_{It}
	  	&=\frac{r_a-r_b+\kappa_{I}\left(1+r_a\right)}{\left(\kappa_{I}-\iota_I\right)\left(1+r_a\right)} x_t
	    \end{align}
	is decreasing in $r_b$ as well as increasing in $r_a$ and $\iota_I$, and converges monotonically to the spot price as the premium rate $\kappa_{I}\to\infty$.
\end{proposition}

\begin{corollary}
If the interest factor
\begin{align}
	\iota_{It}=\frac{r_{bt}-r_{at}}{1+r_{at}}<\kappa_{It}
\end{align}
then the perpetual inverse futures price is equal to the spot price. In this case, the one period cash flow of a long position can be replicated by borrowing 
\begin{align}
  m_{It}x^*_t= \frac{x^*_t}{1+r_{at}}
\end{align}
units of $b$ at rate $r_{bt}$ and investing $m_{It}$ units of $a$ at rate $r_{at}$.
\end{corollary}

% \begin{theorem}
% Assume that \eqref{eq:integrability inverse  DT} holds and that $-1<\iota_{It}<\kappa_{It}$. Then
% 	  \begin{align}
% 	  	\frac1{f_{It}}= E^{Q_b}_t\left[\left(\prod_{\tau=t}^{\theta_{It}-1} \frac{1}{1+\iota_{I\tau}}\right) x^*_{\theta_{It}}\right]
% 	  \end{align}
% 	 where $\theta_{It}\geq t$ is a random time that is defined on an extension $\widetilde\Omega$ of the probability space and distributed according to
% 	  \begin{align}
% 		  Q_{b}\left(\theta_{It}=\sigma\middle|\F\right)
% 		  =\ind{t\leq \sigma}\frac{\kappa_{I\sigma}-\iota_{I\sigma}}{1+\iota_{I\sigma}}\left(\prod_{\tau=t}^{\sigma} \frac{1+\iota_{I\tau}}{1+\kappa_{I\tau}}\right).
% 	  \end{align}
% In particular, if the premium rate is constant and the interest factor is zero then the perpetual inverse futures price is given by $1/f_{It}=E^{Q_b}_t[x^*_{t+\eta_I}]$
% where $\eta_I:\widetilde\Omega\to\N$ is geometrically distributed with mean $1/\kappa_{I}$.
% \end{theorem}

\begin{figure}[t!]
	\centering
	\input{AHJ-Figure3.tex}\label{fig:inverse futures price}	
\end{figure}		

Comparing the above results to those of Section \ref{sec:linear_futures_pricing-DT} shows that the inverse  price behaves much like the linear price. In particular, Proposition \ref{proposition:perpetual inverse futures price constant parameters DT} and Figure \ref{fig:inverse futures price} confirm that in the benchmark case with constant interest rates and funding parameters the inverse futures price 
\begin{align}
  f_{It}=\left[1+\frac{\delta}{\kappa_I\left(1+r_a\right)}\right]x_t
\end{align}
is above the spot if and only if the interest spread $\delta>0$ and converges monotonically to the spot as either $\kappa_I\to\infty$ or $|\delta|\to 0$. The only qualitative difference between the linear and the inverse price is that the latter is increasing in the interest factor whereas the former is decreasing, but this mechanically results from the fact that the interest factor applies to $x_t$ in the linear case but to $x^*_t$ in the inverse case.  

% section linear_futures_pricing (end)

\part{Continuous-time}

\section{The model} % (fold)
\label{sec:the_continuous_time_model}

Time is continuous and indexed by $t\geq 0$. Uncertainty is represented by a filtered probability space $(\Omega,\F,\FF,P)$ where the filtration is right continuous and such that $\F=\cap_{t\geq 0}\F_t$. Unless specified otherwise all processes to appear in what follows are assumed to be adapted to the filtration $\FF$.

As in discrete-time, there are two currencies $i\in\{a,b\}$ and we denote by $x_t$ the $b/a$ exchange rate at date $t$. Investors can freely exchange currencies at this rate and are allowed to invest in two locally riskless assets: one denominated in units of $a$ and the other in units of $b$. The price of these assets satisfy
\begin{align}
  dB_{it}=r_{it} B_{it} dt,\qquad B_{i0}=1, \qquad i\in\{a,b\}
\end{align}
where $r_{it}>-1$ captures the return on the $i-$denominated asset over an infinitesimal time interval starting at date $t$. To ensure the absence of arbitrages between these primitive assets we assume that there exists a probability $Q_a$ that is equivalent to $P$ when restricted to $\F_t$ for any finite $t$ and such that \eqref{eq:NA Qa DT} holds for all real $t\leq s$; and we note that, as in the discrete-time formulation, this assumption implies the existence of a probability $Q_b$ that is equivalent to $P$ when restricted to $\F_t$ for any finite $t$ and such that \eqref{eq:NA Qb DT} is satisfied for all real $t\leq s$.

% section the_discrete_time_model (end)

\section{Linear futures pricing} % (fold)
\label{sec:linear_futures_pricing}

Let $\tau\leq \sigma$ be stopping times. In view of the single period cash flow in \eqref{eq:CF one period} it is clear that the \emph{cumulative} $a-$discounted cash flows from holding a long position in a perpetual futures contract over  $[\![\tau,\sigma]\!]$ are given by
\begin{align}
& \int_{\tau}^\sigma \frac{dC_t}{B_{at}}=  \int_{\tau}^\sigma \frac{1}{B_{at}}\left(df_t- \left(\iota_t x_t+\kappa_t\left(f_t-x_t\right)\right)dt\right)
\end{align}
where the premium rate $\kappa_t>0$ and the interest factor $\iota_t$ are set by the exchange. 
Since a position in the contract can be opened or closed without cost at any point in time, the absence of arbitrage requires that 
\begin{align}\label{eq:NA perpetual CT}
  E^{Q_a}_\tau  \left[\int_{\tau}^\sigma \frac{dC_t}{B_{at}}\right]=0, \qquad \forall \tau\leq \sigma\in\S
\end{align} 
where $\S$ denotes the set of stopping times of the filtration. Rearranging this equality shows that the process
\begin{align}
 M_t:=\int_{0}^t \frac{dC_s}{B_{as}}
\end{align}
is a martingale under $Q_a$. This in turn implies that $C_t$ is a local martingales under $Q_a$ and substituting the definition of the cumulative cash flow process reveals that the perpetual futures prices evolves according to
\begin{align}
  df_t=B_{at}dM_t+\left(\iota_t x_t+\kappa_t\left(f_t-x_t\right)\right)dt.
\end{align}
Finally, discounting at the premium rate $\kappa_t$ on both sides and integrating the resulting expression shows that
\begin{align}\label{eq:pricing restriction in CT}
  e^{-\int_0^t \kappa_u du}f_t+\int_0^t e^{-\int_0^s \kappa_u du}\left(\iota_s-\kappa_s\right) x_s ds
  \in \Mloc(Q_a)  
  % =f_0+\int_0^t e^{-\int_0^s \kappa_u du} B_{as} dM_s
\end{align}
where $\Mloc(Q_a)$ denotes the set of all $Q_a-$local martingales. As in discrete-time, one cannot pin down a unique solution without imposing additional constraints. In particular, if $f_t$ is a solution to \eqref{eq:pricing restriction in CT} then 
\begin{align}
	\hat f_t(\beta)\equiv f_t+ e^{\int_0^t \kappa_u du} \beta_t
\end{align}
is also a solution for any process $\beta\in\Mloc(Q_a)$. We circumvent this difficulty by imposing a \emph{no-bubble condition} which here takes the form:
\begin{align}\label{eq:TVC in CT}
  \lim_{n\to\infty} E^{Q_a}_{t}\left[e^{-\int_0^{\sigma_n} \kappa_u du}f_{\sigma_n}\right]=0, \qquad \forall t\geq 0\text{ and }(\sigma_n)_{n=1}^\infty\subset\S\text{ s.t. }\sigma_n\uparrow\infty. 
\end{align}
Classical arguments then lead to the following continuous-time versions of Theorem \ref{theorem:perpetual futures price DT-1}, Proposition \ref{proposition:perpetual futures price constant parameters DT}, and Corollary \ref{corollary:perpetual futures price equalize DT}. 

\begin{theorem}\label{theorem:futures price CT-1}
  Assume that 
  \begin{align}
    E^{Q_a}\left[ \int_0^\infty e^{-\int_0^{s} \kappa_u du}\left|\kappa_s-\iota_s\right|x_sds\right]<\infty.\label{eq:integrability CT}
  \end{align}
  Then 
  \begin{align}\label{eq:future price CT}
    f_t=E^{Q_a}_{t}\left[ \int_t^\infty e^{-\int_t^{s} \kappa_u du}\left(\kappa_s-\iota_s\right)x_sds\right]
  \end{align}
   is the unique process that satisfies \eqref{eq:pricing restriction in CT} and \eqref{eq:TVC in CT}. If in addition $\iota_t<\kappa_t$ then this process can be represented as
	  \begin{align}
	  	f_t= E^{Q_a}_t\left[e^{-\int_t^{\theta_t} \iota_u du}  x_{\theta_t}\right]
	  \end{align}
    where $\theta_t$ is distributed according to
    \begin{align}
          Q_{a}\left(\theta_t \in ds\middle| \F\right)
          =\ind{s\geq t}e^{-\int_t^s (\kappa_u-\iota_u)du}\left(\kappa_s-\iota_s\right) ds, \qquad s\geq 0.
    \end{align}
    In particular, if the premium rate $\kappa$ is constant and the interest factor $\iota\equiv 0$ then the perpetual futures price is given by 
	  \begin{align}
	  	f_t=E^{Q_a}_t \left[\int_{t}^{\infty} \kappa e^{-\kappa (s-t) }x_s ds\right]=E^{Q_a}_t\left[x_{t+\tau}\right]
	  \end{align}
	  where $\tau$ is an exponentially distributed random time with mean $1/\kappa$.
\end{theorem}

\begin{proposition}\label{proposition:constant coefficients CT}
	If $\iota<\kappa$ and $(r_a,r_b)$ are constants such that $\kappa+r_b-r_a>0$ then
    \begin{align}\label{eq:price CT-constant}
    f_t=\frac{\kappa-\iota}{\kappa+r_b-r_a} x_t
    \end{align}
is increasing in $r_a$ as well as decreasing in $r_b$ and $\iota$, and converges monotonically to the spot price as the premium rate $\kappa\to\infty$.
\end{proposition}

\begin{corollary}\label{corollary:equalizing CT}
Assume that 
\begin{align}
	\iota_t\equiv r_{at}-r_{bt}< \kappa_t
\end{align}
then the perpetual futures price is equal to the spot price at all times and can be dynamically replicated by trading in the two riskless assets.
\end{corollary}

\noindent
While the proofs are slightly more involved in continuous-time, we stress that the results and their implications are essentially the same in both frameworks. The only difference is that $(r_a,r_b,\iota)$ are simply compounded period rates in discrete-time but continuously compounded annual rates in continuous-time.
%

% section linear_futures_pricing (end)

\section{Inverse and quanto futures pricing} % (fold)
\label{sec:inverse_futures_pricing}

\subsection{Perpetual inverse futures} % (fold)
\label{sub:perpetual_inverse_futures}

Let now $x^*_t:=1/x_t$ denote the price of one unit of currency $a$ in units of currency $b$ and recall that we denote by $Q_b$ the pricing measure for $b-$denominated amounts. Proceeding as in the case of the linear contract shows that the arbitrage restriction for the inverse contract is given by
\begin{align}
  E^{Q_b}_\tau  \left[\int_{\tau}^\sigma \frac{dC_{it}}{B_{bt}}\right]=0, \qquad \forall \tau\leq \sigma\in\S,
\end{align} 
where the incremental cash flow
\begin{align}
  dC_{it}=d \left(\frac{1}{f_{It}}\right)-\kappa_{It}\left(\frac{1}{f_{It}}-x^*_t\right)dt-\iota_{It} x^*_t dt
\end{align}
is denominated in units of $b$ and the parameters $(\iota_{It},\kappa_{It})$ are set by the exchange subject to $\kappa_{It}>0$. The same arguments as in the linear case then imply that
\begin{align}\label{eq:pricing restriction inverse in CT}
  e^{-\int_0^t \kappa_{Iu} du}\frac{1}{f_{It}}+\int_0^t e^{-\int_0^s \kappa_{Iu} du}\left(\iota_{Iu}-\kappa_{Iu}\right) x^*_s ds
  \in  \Mloc(Q_b)
  % =\frac{1}{f_{I0}}+\int_0^t{e^{-\int_0^s \kappa_{Iu} du} B_{bs}dM_{is}}
\end{align}
and imposing the no-bubble condition
\begin{align}\label{eq:TVC inverse in CT}
  \lim_{n\to\infty} E^{Q_b}_{t}\left[e^{-\int_0^{\sigma_n} \kappa_{Iu} du}\frac{1}{f_{I\sigma_n}}\right]=0, \qquad \forall t\geq 0\text{ and }(\sigma_n)_{n=1}^\infty\subset\S\text{ s.t. }\sigma_n\uparrow\infty
\end{align}
allows to uniquely determine the inverse futures price. In particular, we have the following counterparts to Theorem \ref{theorem:futures price CT-1}, Proposition \ref{proposition:constant coefficients CT}, and Corollary \ref{corollary:equalizing CT}.

\begin{theorem}
  Assume that 
  \begin{align}
    E^{Q_b}\left[ \int_0^\infty e^{-\int_0^{s} \kappa_{Iu} du}\left|\kappa_{Is}-\iota_{Is}\right|x^*_sds\right]<\infty.\label{eq:integrability inverse CT}
  \end{align}
  Then the process
  \begin{align}
    \frac{1}{f_{It}}=E^{Q_b}_{t}\left[ \int_t^\infty e^{-\int_t^{s} \kappa_{Iu} du}\left(\kappa_{Is}-\iota_{Is}\right)x^*_sds\right]
  \end{align}
   is the unique solution to \eqref{eq:pricing restriction inverse in CT} that satisfies \eqref{eq:TVC in CT}. If in addition $\iota_{It}<\kappa_{It}$ then this solution can be represented as	  
   \begin{align}
	  	\frac1{f_{It}}= E^{Q_b}_t\left[e^{-\int_t^{\theta_{It}} \iota_{Iu} du}  x^*_{\theta_{It}}\right]
   \end{align}
	  where $\theta_{It}$ is distributed according to
	  \begin{align}
		  Q_{b}\left(\theta_{It}\in ds\middle| \F\right)
		  =\ind{t\leq s}e^{-\int_t^s (\kappa_{Iu}-\iota_{Iu})du}\left(\kappa_{Is}-\iota_{Is}\right)ds,\qquad s\geq 0.
	  \end{align}
    In particular, if the premium rate $\kappa_{I}$ is constant and the interest factor $\iota_I\equiv 0$ then the perpetual inverse futures price is given by 
	  \begin{align}
	  	\frac1{f_{It}}=E^{Q_b}_t \left[\int_t^{\infty} \kappa_{I} e^{-\kappa_{I} (s-t) }x^*_s ds\right]=E^{Q_b}_t\left[x^*_{t+\tau_I}\right]
	  \end{align}
	  where $\tau_I$ is an exponentially distributed random time with mean $1/\kappa_{I}$.
\end{theorem}

\begin{proposition}
	If $\iota_I<\kappa_{I}$ and $(r_a,r_b)$ are constants such that $\kappa_{I}+r_a-r_b>0$
  then the perpetual inverse futures price
    \begin{align}
    f_{It}=\frac{\kappa_{I}+r_a-r_b}{\kappa_{I}-\iota_I}x_t
    \end{align}
is decreasing in $r_b$ as well as increasing in $r_a$ and $\iota$, and converges monotonically to the spot price as the premium rate $\kappa_{I}\to\infty$.
\end{proposition}

\begin{corollary}
If the interest factor
\begin{align}
	\iota_{It}= r_{bt}-r_{at}<\kappa_t
\end{align}
then the perpetual inverse futures price equals the spot price at all times and can be dynamically replicated by trading in the two riskless assets.
\end{corollary}

% subsection perpetual_inverse_futures (end)

\subsection{Perpetual quanto futures} % (fold)
\label{sub:perpetual_quanto_futures}

Let now $c$ denote a \emph{third} currency. \emph{Quanto} futures are contracts that give exposure to the $c/a$ exchange rate $z_t>0$ but are quoted, margined, and funded in currency $b$. Specifically, the size of the contract is set to one unit of $c$ and the periodic cash flows are computed in units of $a$ but paid in units $b$ after conversion using a fixed $a/b$ exchange rate $\chi^*$ specified in the contract. For example, the \href{https://www.bitmex.com/app/trade/ETHUSD}{BTC$|$ETH$|$USD} contract traded on the platform BitMEX is a perpetual futures contract on the ETH/USD exchange rate that operates in BTC using a fixed USD/BTC exchange rate. See \citet{BitMEX-Q} and \citet{alexander2023} for a discussion of the distinction between inverse and quanto derivatives.

The instantaneous cash flow \emph{in units of $b$} from a long position in the perpetual quanto futures contract is
\begin{align}
 dC_{qt}= \chi^* df_{qt}-\chi^*\left(\iota_{qt} z_t +\kappa_{qt}\left({f_{qt}}-z_t\right)\right)dt
\end{align}
where the funding parameters $(\iota_{qt},\kappa_{qt})$ are set by the exchange subject to $\kappa_{qt}>0$, and the absence of arbitrage requires that
\begin{align}
 \chi^* E^{Q_b}_\tau  \left[\int_{\tau}^\sigma \frac{dC_{qt}}{B_{bt}}\right]=0, \qquad \forall \tau\leq \sigma\in\S.
\end{align} 
Proceeding as in the linear case then shows that
\begin{align}\label{eq:pricing restriction quanto in CT}
  e^{-\int_0^t \kappa_{qu} du}f_{qt}+\int_0^t e^{-\int_0^s \kappa_{qu} du}\left(\iota_{qu}-\kappa_{qu}\right) z_s ds
  \in\Mloc(Q_b),
\end{align}
and imposing the no-bubble condition
\begin{align}\label{eq:TVC quanto in CT}
\lim_{n\to\infty} E^{Q_b}_{t}\left[e^{-\int_0^{\sigma_n} \kappa_{qu} du}{f_{q\sigma_n}}\right]=0, \qquad \forall t\geq 0\text{ and }(\sigma_n)_{n=1}^\infty\subset\S\text{ s.t. }\sigma_n\uparrow\infty
\end{align}
allows to uniquely determine the quanto futures price. 

\begin{theorem}
  Assume that 
  \begin{align}
    E^{Q_b}\left[ \int_0^\infty e^{-\int_0^{s} \kappa_{qu} du}\left|\kappa_{qs}-\iota_{qs}\right|z_sds\right]<\infty.\label{eq:integrability quanto CT}
  \end{align}
  Then the process
  \begin{align}
    f_{qt}=E^{Q_b}_{t}\left[ \int_t^\infty e^{-\int_t^{s} \kappa_{qu} du}\left(\kappa_{qs}-\iota_{qs}\right)z_sds\right]
  \end{align}
   is the unique solution to \eqref{eq:pricing restriction quanto in CT} that satisfies \eqref{eq:TVC quanto in CT}. If, furthermore, $\kappa_{qt}-\iota_{qt}>0$ then this solution can be represented as
   \begin{align}
   E^{Q_b}_t\left[e^{-\int_t^{\theta_{qt}} \iota_{qu} du}  z_{\theta_{qt}}\right]
  \end{align}
  where $\theta_{qt}$ is distributed according to
  \begin{align}
	  Q_{b}\left(\theta_{qt}\in ds \middle| \F\right)
	  =\ind{t\leq s}e^{-\int_t^s (\kappa_{qu}-\iota_{qu})du}\left(\kappa_{qs}-\iota_{qs}\right)ds,\qquad s\geq 0.
  \end{align}
   In particular, if the premium rate $\kappa_{q}$ is constant and the interest factor $\iota\equiv 0$ then the perpetual quanto futures price is given by 
  \begin{align}\label{eq:everlasting option as expectation}
  	f_{qt}=E^{Q_b}_t \left[\int_t^{\infty} \kappa e^{-\kappa (s-t) }z_s ds\right]=E^{Q_b}_t\left[z_{t+\tau_q}\right]
  \end{align}
  where $\tau_q$ is an exponentially distributed random time with mean $1/\kappa$.
\end{theorem}

In contrast to the linear and inverse cases it is no longer possible to obtain explicit expressions for the futures price or  the equalizing interest factor without specifying a model for the exchange rates $z_t$ and $x_t$. This difficulty arises from the fact that the exchange rate $z_t$ applies to the currency pair $c/a$ and, thus, is unrelated to currency $b$. As a result, its drift under $Q_b$ depends on the covariance between changes in $z_t$ and changes in the $b/a$ exchange rate $x_t$, and this covariance can only be computed once we specify the sources of risk that affect theses processes.

To illustrate the required convexity adjustment in a simple setting, assume that the interest rates $\left(r_a,r_b,r_c,\iota,\kappa\right)$ are constants and that the exchange rates $\left(x_t,z_t\right)$ evolves according to
\begin{align}
  dx_t/x_t
  =\left(r_a-r_b\right)dt + \sigma_{x}^* dZ_{at},\label{eq:CT dynamics Black-Scholes}\\
  dz_t/z_t
  =\left(r_a-r_c\right)dt + \sigma_{z}^* dZ_{at},
\end{align}
where $r_c$ is the constant interest rate that applies to $c-$denominated riskfree deposits and loans, $(\sigma_{x},\sigma_z)$ are constant vectors of dimension $n$, and $Z_{at}$ is a $n-$dimensional Brownian under the probability measure $Q_a$. Since
\begin{align}
  \left.\frac{dQ_b}{dQ_a}\right|_{\F_t}=e^{(r_b-r_a)t}\frac{x_t}{x_0}=e^{-\frac12 \|\sigma_x\|^2 t+\sigma_x^* Z_{at}}
\end{align}
it follows from Girsanov's theorem that 
\begin{align}
dz_t/z_t=\left(r_a-r_c+\sigma_{x}^*\sigma_{z}\right)dt + \sigma_{z}^* dZ_{bt}
\end{align}
where $Z_{bt}$ is an $n-$dimensional Brownian under $Q_b$ and using this evolution allows to derive the following result:

\begin{proposition}
Assume that $\iota<\kappa$ and $r_c-\sigma_{x}^*\sigma_{z}-r_a+\kappa>0$ then
  \begin{align}
  f_{qt}=\frac{\left(\kappa-\iota\right)z_t}{r_c-\sigma_{x}^*\sigma_{z}-r_a+\kappa} 
  \end{align}
is decreasing in $r_c$ and $\iota$ as well as increasing in $r_a$ and $\sigma_{x}^*\sigma_{z}$, and converges monotonically to the spot price as the premium rate $\kappa\to\infty$. 
\end{proposition}

% subsection perpetual_quanto_futures (end) 

% section linear_futures_pricing (end)

\section{Everlasting options} % (fold)
\label{sec:everlasting_futures_pricing}

Everlasting options work in a manner similar to perpetual futures but track a function of the spot price instead of the spot price itself. The idea and the name seem to go back to \citet{paradigm2021everlasting} and some specifications, such as quadratic perpetuals, are traded on some exchanges. See, for example \citet{medium2022sqeeth}. Note, however, that the name can be misleading: everlasting options are not actually options but perpetual future contracts. In particular, an everlasting call or put option is very different from a perpetual American option.

Let $\varphi:\R_+\to \R$ denote a payoff function. The futures price for the corresponding everlasting option is quoted and margined in units of $a$. As a result, the cumulative cash flows from holding a long position evolves according to
\begin{align}
 dC_{ot}=  df_{ot}-\kappa_{ot}\left(f_{ot}-\varphi(x_t)\right)dt
\end{align}
where $f_{ot}$ denotes the (futures) price of the everlasting option and $\kappa_{ot}>0$ is a premium rate set by the exchange.
The same arguments as in the linear case then imply that the absence of arbitrage requires that
\begin{align}\label{eq:pricing restriction eo in CT}
  e^{-\int_0^t \kappa_{Iu} du}f_{ot}+\int_0^t e^{-\int_0^s \kappa_{Iu} du}\kappa_{os}\varphi\left(x_s\right) ds
  \in  \Mloc(Q_a)
\end{align}
and imposing the no-bubble condition
\begin{align}\label{eq:TVC eo in CT}
  \lim_{\sigma\to\infty} E^{Q_a}_{t}\left[e^{-\int_0^{\sigma_n} \kappa_{ou} du} f_{o \sigma_n}\right]=0, \qquad \forall t\geq 0\text{ and }(\sigma_n)_{n=1}^\infty\subset\S\text{ s.t. }\sigma_n\uparrow\infty
\end{align}
allows to uniquely determine the inverse futures price. 

\begin{theorem}\label{theorem:everlasting option pricing}
  Assume that 
  \begin{align}
    E^{Q_a}\left[ \int_0^\infty e^{-\int_0^{s} \kappa_{ou} du}\kappa_{os} \varphi(x_s)ds\right]<\infty.\label{eq:integrability eo CT}
  \end{align}
  Then the process
  \begin{align}
    f_{ot}=E^{Q_a}_{t}\left[\int_t^\infty e^{-\int_t^{s} \kappa_{ou} du}\kappa_{os}\varphi(x_s) ds\right]
  \end{align}
   is the unique solution to \eqref{eq:pricing restriction eo in CT} that satisfies condition \eqref{eq:TVC eo in CT}. Furthermore, this unique solution can be expressed as 
   \begin{align}\label{eq:representation everlasting o}
  f_{ot}=  E^{Q_a}_t\left[\varphi\left( x_{\theta_{ot}}\right)\right]
  \end{align}
  where $\theta_{ot}\geq t$ is a random time that is defined an extension of the probability space and distributed according to
  \begin{align}
	  Q_{a}\left(\theta_{ot}\in ds\middle| \F\right)
	  =\ind{t\leq s}e^{-\int_t^s \kappa_{ou}du}\kappa_{os} ds,\qquad s\geq 0.
  \end{align}
   In particular, if $\kappa_{o}$ is constant then 
  \begin{align}
  	f_{ot}=E^{Q_a}_t \left[\int_t^{\infty} \kappa e^{-\kappa (s-t) }\varphi\left( x_s\right) ds\right]=E^{Q_a}_t\left[\varphi\left( x_{t+\tau_o}\right)\right]
  \end{align}
  where $\tau_o$ is an exponentially distributed random time with mean $1/\kappa$.
\end{theorem}

\citet*{angeris2023primer} study an inverse problem that is related to Theorem \ref{theorem:everlasting option pricing}. However, their results are not directly comparable to ours because the contract that they analyze is different.  Indeed, they consider an $a-$denominated perpetual contract in which the long pays the short some amount $\Phi_\tau$ \emph{upon entering the contract} at date $\tau$ as well as funding at some rate $F_t dt$ as long as the position is opened, and receives the amount $\Phi_\sigma$ upon exiting the contract at date $\sigma\geq \tau$. In this product specification, the cumulative net discounted cash flow of a long position over the holding period is
\begin{align}
	N:=\frac{\Phi_\sigma}{B_{a\sigma}} -\frac{\Phi_\tau}{B_{a\tau}}-\int_\tau^\sigma {\frac{F_s}{B_{as}} ds}=
	\int_\tau^\sigma d\left(\frac{\Phi_s}{B_{as}}\right)-\int_\tau^\sigma {\frac{F_s}{B_{as}} ds}.
\end{align}
By contrast, in our formulation the cumulative discounted cash flows from holding a long position in the corresponding perpetual futures would be
\begin{align}
	\int_\tau^\sigma \left( \frac{d\Phi_s}{B_{as}}-\frac{F_s}{B_{as}} ds\right)
	=N+\int_\tau^\sigma  r_{as}\frac{\Phi_s}{B_{as}}ds
\end{align}
because the contract is margined continuously rather than only at the beginning and end. We believe that our formulation more closely reflects the actual functioning of markets because in practice contracts do not require an upfront payment. 

Within our formulation, the problem analyzed by \citet{angeris2023primer} consists in determining the funding rate $F_t$ in such a way that the associated perpetual futures price $\Phi_t=\varphi(x_t)$ at all times for some given function.  If the exchange rate follows an It\^o process as in their paper, then \eqref{eq:NA Qa DT} implies that
\begin{align}
dx_t=\left(r_{at}-r_{bt}\right) x_t dt +  \Sigma_{xt}^* dZ_{at}
\end{align}
for some diffusion coefficient $\Sigma_{xt}$ and some $Q_a-$Brownian motion $Z_{at}$ of the same dimension.
In this case, the problem can be solved by a direct application of It\^o's lemma. Indeed, if $\varphi\in C^2$ then the funding rate
\begin{align}
	F_t(\varphi):=\varphi'(x_t)\left(r_{at}-r_{bt}\right) x_t+\frac12 \varphi''(x_t) \|\Sigma_{xt}\|^2
\end{align}
has the property that 
\begin{align}
	\int_{0}^t \frac{1}{B_{as}}\left(d \varphi(x_s)- F_s(\varphi)\right)ds
	\in \Mloc(Q_a).
\end{align}
Under appropriate integrability assumptions on $\varphi(x_t)$ this local martingale property in turn implies that
\begin{align}
	E^{Q_a}_\tau\left[\int_{\tau}^\sigma \frac{1}{B_{as}}\left(d \varphi(x_s)- F_s(\varphi)\right)ds\right]=0, \qquad \forall \tau\leq \sigma\in\S,
\end{align} 
and it follows that $\varphi(x_t)$ is the futures price induced by the funding rate $F_t(\varphi)$. In particular, for the identity function we find that
\begin{align}
	F_t(\varphi)=\left(r_{at}-r_{bt}\right) x_t
\end{align}
from which we recover the result of Corollary \ref{corollary:equalizing CT}. Note however that the latter corollary was proved without any assumption on the evolution of the exchange rate other than the no-arbitrage condition.

To illustrate the pricing of everlasting options, assume the exchange rate evolves according to \eqref{eq:CT dynamics Black-Scholes} and consider the call and put options with premium rate $\kappa>0$ and strike $K$. The associated prices are given by
\begin{align}
     {c_t}&= E^{Q_a}_t\left[\int_t^{\infty} \kappa e^{-\kappa(s-t)} \left({x}_s-K \right)^+ds  \right],\\
     {p_t}&= E^{Q_a}_t\left[\int_t^{\infty} \kappa e^{-\kappa(s-t)} \left({K}-x_s \right)^+ds  \right],
\end{align}
and satisfy the everlasting put-call parity
\begin{align}
        {c_t-p_t}&=E^{Q_a}\left[\int_t^{\infty} \kappa e^{-\kappa(s-t)} \left(x_s-K \right)ds  \right]=f(x_t)-K
\end{align}
where 
\begin{align}
 f(x_t):=E^{Q_a}_t\left[\int_t^{\infty} \kappa e^{-\kappa(s-t)} x_s ds  \right]=\frac{ \kappa x_t}{\kappa-r_a+r_b}
\end{align}
denotes the perpetual futures price with premium rate $\kappa$ and interest factor $\iota=0$.
Furthermore, a standard calculation based on properties of the geometric Brownian motion delivers an explicit formula for the everlasting option prices:

\begin{proposition} \label{proposition:everlasting_call_put_prices}
Assume that $\kappa-r_a+r_b>0$. Then the Black-Scholes-Merton prices of the everlasting call and put are explicitly given by
\begin{align}
  p_t&=c_t+K-f(x_t),\\[0.1cm]
  % \intertext{and}
  c_t&=\begin{cases}
     x_t^\Theta K^{1-\Theta} \left(\frac{\Pi\left(r_a - r_b\right) - \kappa}{\left(\Pi-\Theta\right)\left(\kappa-r_a+r_b \right)}\right), & x_t\leq K\\
     x_t^\Pi K^{1-\Pi} \left(\frac{\Theta\left(r_a - r_b\right) - \kappa}{\left(\Pi-\Theta\right)\left(\kappa-r_a+r_b\right)}\right) +  f(x_t) - K, & x_t>K
  \end{cases}
\end{align}
where $\Pi<0$ and $\Theta>1$ are the roots of the quadratic equation
\begin{align} \label{eq:quadratic}
\left(r_a-r_b\right) \xi+\frac12 \xi \left(\xi-1\right)\|\sigma_x\|^2- \kappa=0
\end{align}
associated with the dynamics of the $b/a$ exchange rate under $Q_a$.
\end{proposition}
\noindent

\begin{figure}[t!]
	\centering
	\input{AHJ-Figure4.tex}
\end{figure}		

Figure \ref{fig:everlasting call} plots the arbitrage price $c(x)$ and Delta $c'(x)$ of an everlasting call with $\kappa=K=1$ in a model with zero interest rates, and compares these quantities with the price and delta of a European call with maturity $T=1/\kappa=1$. The figure shows that the everlasting price behaves essentially like the European price and confirms that everlasting options provide an effective way for investors to establish a long-term option-like exposure without incurring the risks and costs associated with the rolling of positions in successively maturing options. 

Explicit pricing formulas along these lines can be derived in closed form for any given payoff function $\varphi(x)$ as long as \eqref{eq:integrability eo CT} is satisfied (see \citet{Zervos} for a general formula in an diffusion context).
More generally, and as proposed by \citet{Kristensen}, one may consider everlasting  path dependent options in which the payoff $\varphi(x_t)$ is replaced by a functional $\varphi\left(x_u:u\leq t\right)$ of the path of the underlying. The initial price of such a derivative is 
\begin{align}
\kappa E^{Q_a}\left[\int_0^\infty  e^{-\kappa t}  \varphi\left(x_u:u\leq t\right) dt \right]=\kappa \int_0^\infty  e^{-\kappa t}  E^{Q_a}\left[\varphi\left(x_u:u\leq t\right)\right] dt 
\end{align}
and viewing the right handside as a Laplace transform in the maturity variable of the price of a path-dependent European derivative allows to exploit the host of explicit formulae available for such transforms (see, e.g., \citet{borodin} and \citet{jeanblanc}).
Example of everlasting path dependent options that can be priced in this way include lookbacks, barriers, occupation time derivatives, Parisian derivatives, and continuously sampled Asian options.

% section linear_futures_pricing (end)

\section{Conclusion} % (fold)
\label{sec:conclusion}

In this paper, we address the arbitrage pricing of perpetual currency futures contracts. We study discrete-time as
well as continuous-time formulations, and identify funding specifications for the linear and inverse contracts that
guarantee that the futures price coincides with the spot. In both cases, the required interest factor is a
simple function of the interest rates in the underlying currencies. Under the assumption of constant interest rates and funding parameters, we derive explicit model-free expressions for linear and inverse futures prices. Furthermore, we
show that, in general, the perpetual future price can represented as the discounted expected value
of the underlying price at a random time that reflects the funding specification of the contract. We also derive general continuous-time results for perpetual quanto futures as well as for everlasting options, and illustrate those results with closed form formulae for everlasting calls and puts in a Black-Scholes setting.

There are several directions in which to pursue the analysis of this paper. First, it would be natural to formulate a joint multifactor (e.g., affine) model for the exchange rates, interest rates, and funding parameters that delivers closed form solutions for perpetual futures prices and that can be estimated by relying on observed interest rates, funding parameters, and futures prices. Second, given that perpetual futures are likely to find traction in other classes than (crypto) currencies, it would be fruitful to adapt the pricing formulae of this paper to other asset classes such as equities and commodities. Third, and finally, we believe that it would be interesting to model the dynamic determination of funding parameters by the exchange in response to demand within an equilibrium setting where the price of the underlying, the perpetual futures price, and possibly the interest rates are endogenously determined.

% section conclusion (end)

\appendix

\section{Proofs} % (fold)
\label{sec:proofs}

\subsection{Discrete-time results} % (fold)
\label{sub:discrete_time}

Since the proofs are similar in the linear and inverse cases we only provide complete details for the results pertaining to linear contracts.

\begin{proof}[Proof of Theorem \ref{theorem:perpetual futures price DT-1}] 
  Assume that the process $f_t$ is a solution to \eqref{eq:pricing restriction in DT} that satisfies \eqref{eq:TVC in DT}. Passing to the limit in \eqref{eq:pricing restriction in DT} and using \eqref{eq:TVC in DT} shows that
  \begin{align}
    f_t=\lim_{T\to\infty} E^{Q_a}_{t}\left[ \sum_{\sigma=t}^{T-1} \left(\prod_{\tau=t}^{\sigma} \frac{1}{1+\kappa_\tau}\right)  \left(\kappa_\sigma-\iota_\sigma\right)x_{\sigma}\right]
  \end{align}
  and the conclusion now follows from condition \eqref{eq:integrability DT} by dominated convergence.  Under the additional assumption that $-1<\iota_t<\kappa_t$ we have that
  \begin{align}
  	E^{Q_a}_t&\left[\left(\prod_{\tau=t}^{\theta_t-1} \frac{1}{1+\iota_\tau}\right)  x_{\theta_t}\right]
  	=E^{Q_a}_t\left[\sum_{\sigma=t}^\infty  \left(\prod_{\tau=t}^{\sigma-1} \frac{1}{1+\iota_\tau}\right)  x_{\sigma} \, Q_{a}\left(\theta_t =\sigma \middle| \F\right) \right]\\
  	&=E^{Q_a}_t\left[\sum_{\sigma=t}^\infty  \left(\prod_{\tau=t}^{\sigma-1} \frac{1}{1+\iota_\tau}\right)  x_{\sigma}  \frac{\kappa_\sigma-\iota_\sigma}{1+\iota_\sigma}\left(\prod_{\tau=t}^{\sigma} \frac{1+\iota_\tau}{1+\kappa_\tau}\right) \right]\\
    &=E^{Q_a}_t\left[\sum_{\sigma=t}^\infty  \left(\prod_{\tau=t}^{\sigma} \frac{1}{1+\kappa_\tau}\right) \left(\kappa_\sigma-\iota_\sigma\right) x_{\sigma} \right]=f_t
  \end{align}
  where the last equality follows from the first part of the proof. The second part directly follows from the first when $\kappa$ is constant and $\iota\equiv 0$.
\end{proof}

\begin{proof}[Proof of Proposition \ref{proposition:perpetual futures price constant parameters DT}]
  Under the stated assumption
  \begin{align}
     \sum_{\sigma=0}^{\infty} \left(\prod_{\tau=0}^{\sigma} \frac{1}{1+\kappa}\right)  E^{Q_a}\left[x_{\sigma}\right]
    &= \frac{1}{1+\kappa}\sum_{\sigma=0}^{\infty} \left[\frac{1}{1+\kappa}\left(\frac{1+r_a}{1+r_b}\right)\right]^{\sigma}  x_{0}\\
    &= \frac{1+r_b}{r_b-r_a+\kappa\left(1+r_b\right)} x_{0}<\infty
  \end{align}
  where the second equality follows from the no-arbitrage restriction \eqref{eq:NA Qa DT}. Therefore, condition \eqref{eq:integrability DT} is satisfied and it thus follows from Theorem \ref{theorem:perpetual futures price DT-1} that
  \begin{align}
    f_t
    &=\frac{\kappa-\iota}{1+\kappa}\sum_{\sigma=t}^\infty \left(\frac{1}{1+\kappa}\right)^{\sigma -t}E^{Q_a}_t\left[ x_\sigma\right]\\
    &=\frac{\kappa-\iota}{1+\kappa}\sum_{\sigma=t}^\infty  \left[\frac{1}{1+\kappa}\left(\frac{1+r_a}{1+r_b}\right)\right]^{\sigma -t}x_t=\frac{\left(\kappa-\iota\right)\left(1+r_b\right)}{r_b-r_a+\kappa\left(1+r_b\right)} x_{t}
  \end{align}
  where the second equality also follows from \eqref{eq:NA Qa DT}. The comparative statics follow by differentiating the futures price. 
  
To prove the second part observe that under this specification the cash flow at date $t+1$ from a long position in the perpetual futures
\begin{align}
  x_{t+1}-x_t+ \frac{r_{at}-r_{bt}}{1+r_{bt}} x_t=m_t\left(1+r_{bt}\right)x_{t+1} -m_tx_t\left(1+r_{at}\right) 
\end{align} 
coincides with the outcome 
of a cash and carry trade that borrows $m_tx_t$ units of $a$ at rate $r_{at}$ to buy $m_t$ units of $b$ which are invested at rate $r_{bt}$ until date $t+1$ where the proceeds are converted back to units of $a$ and used to payback the loan.
\end{proof}

\begin{proof}[Proof of Corollary \ref{corollary:perpetual futures price equalize DT}]
  Under the stated assumptions
  \begin{align}
     & E^{Q_a}  \left[ \sum_{\sigma=0}^{\infty} \left(\prod_{\tau=0}^{\sigma} \frac{1}{1+\kappa_\tau}\right) \left|\kappa_\sigma-\iota_\sigma\right| x_{\sigma}\right]\\
             &=     E^{Q_a} \left[ \sum_{\sigma=0}^{\infty} \left(\prod_{\tau=0}^{\sigma-1} \frac{1}{1+\kappa_\tau}\frac{1+r_{a\tau}}{1+r_{b\tau}}\right) \frac{r_{b\sigma}-r_{a\sigma}+ \kappa_\sigma\left(1+r_{b\sigma}\right)}{\left(1+\kappa_\sigma\right)\left(1+r_{b\sigma}\right)}  \left(\frac{B_{b\sigma}}{B_{a\sigma}} x_{\sigma}\right)\right]\\
             &=     E^{Q_b} \left[ \sum_{\sigma=0}^{\infty} \left(\prod_{\tau=0}^{\sigma-1} \frac{1}{1+\kappa_\tau}\frac{1+r_{a\tau}}{1+r_{b\tau}}\right) \frac{r_{b\sigma}-r_{a\sigma}+ \kappa_\sigma\left(1+r_{b\sigma}\right)}{\left(1+\kappa_\sigma\right)\left(1+r_{b\sigma}\right)}  \right] x_0 
             =x_0<\infty
  \end{align}
  where the second equality follows from the fact that absent arbitrage opportunities the pricing measures are related by
  \begin{align}\label{eq:dQa/dQb in DT}
    \left.\frac{dQ_a}{dQ_b}\right|_{\F_t}=\frac{B_{bt}}{B_{at}} \frac{x_{t}}{x_0},\qquad t\geq 0.
  \end{align}
  Therefore, condition \eqref{eq:integrability DT} is satisfied and Theorem \ref{theorem:perpetual futures price DT-1} shows that the perpetual futures price is given by
  \begin{align}
    f_t
    &=E^{Q_a}_t \left[ \sum_{\sigma=t}^{\infty} \left(\prod_{\tau=t}^{\sigma-1} \frac{1}{1+\kappa_\tau}\frac{1+r_{a\tau}}{1+r_{b\tau}}\right) \frac{r_{b\sigma}-r_{a\sigma}+ \kappa_\sigma\left(1+r_{b\sigma}\right)}{\left(1+\kappa_\sigma\right)\left(1+r_{b\sigma}\right)}  \left(\frac{B_{b\sigma}B_{at}}{B_{a\sigma}B_{bt}} x_{\sigma}\right)\right]\\
    &=E^{Q_b}_t \left[ \sum_{\sigma=t}^{\infty} \left(\prod_{\tau=t}^{\sigma-1} \frac{1}{1+\kappa_\tau}\frac{1+r_{a\tau}}{1+r_{b\tau}}\right) \frac{r_{b\sigma}-r_{a\sigma}+ \kappa_\sigma\left(1+r_{b\sigma}\right)}{\left(1+\kappa_\sigma\right)\left(1+r_{b\sigma}\right)} \right] x_t=x_t
  \end{align}
  where the second equality follows from \eqref{eq:dQa/dQb in DT}.
\end{proof}

% \begin{proof}[Proof of Theorem \ref{theorem:perpetual futures price DT-2}]
%   Under the stated assumption
%   \begin{align}
%   	E^{Q_a}_t&\left[\left(\prod_{\tau=t}^{\theta_t-1} \frac{1}{1+\iota_\tau}\right)  x_{\theta_t}\right]
%   	=E^{Q_a}_t\left[\sum_{\sigma=t}^\infty  \left(\prod_{\tau=t}^{\sigma-1} \frac{1}{1+\iota_\tau}\right)  x_{\sigma} \, Q_{a}\left(\theta_t =\sigma \middle| \F\right) \right]\\
%   	&=E^{Q_a}_t\left[\sum_{\sigma=t}^\infty  \left(\prod_{\tau=t}^{\sigma-1} \frac{1}{1+\iota_\tau}\right)  x_{\sigma}  \frac{\kappa_\sigma-\iota_\sigma}{1+\iota_\sigma}\left(\prod_{\tau=t}^{\sigma} \frac{1+\iota_\tau}{1+\kappa_\tau}\right) \right]\\
%     &=E^{Q_a}_t\left[\sum_{\sigma=t}^\infty  \left(\prod_{\tau=t}^{\sigma} \frac{1}{1+\kappa_\tau}\right) \left(\kappa_\sigma-\iota_\sigma\right) x_{\sigma} \right]=f_t
%   \end{align}
%   where the last equality follows from Theorem \ref{theorem:futures price CT-1}. The second part directly follows from the first when $\kappa$ is constant and $\iota\equiv 0$.
% \end{proof}

% subsection discrete_time (end)

\subsection{Continuous-time results} % (fold)
\label{sub:continuous_time}

Since the proofs are similar in the linear, inverse, and quanto cases we only provide complete details for the results pertaining to linear contracts.

\begin{proof}[Proof of Theorem \ref{theorem:futures price CT-1}]
Assume that the process $f_t$ satisfies both \eqref{eq:pricing restriction in CT} and \eqref{eq:TVC in CT}. Due to \eqref{eq:pricing restriction in CT} we have that 
\begin{align}
	e^{-\int_0^t \kappa_u du}f_t+\int_0^t e^{-\int_0^s \kappa_u du}\left(\iota_s-\kappa_s\right) x_s ds\in \Mloc(Q_a).
\end{align}
Therefore, it follows that for any given date $t\geq 0$ there exists a sequence of stopping times $(\sigma_n)_{n=1}^\infty$ such that $t\leq \sigma_n\uparrow\infty$ and
\begin{align}
	f_t=E^{Q_a}_t \left[ e^{-\int_t^{\sigma_n} \kappa_u du}f_{\sigma_n}+\int_t^{\sigma_n} e^{-\int_t^s \kappa_u du}\left(\kappa_s-\iota_s\right) x_s ds\right], \qquad n\geq 1.
\end{align}
Letting $n\to\infty$ on both sides of this equality and using the no-bubble condition \eqref{eq:TVC in CT} then shows that we have
\begin{align}
	f_t=\lim_{n\to\infty} E^{Q_a}_t \left[ \int_t^{\sigma_n} e^{-\int_t^s \kappa_u du}\left(\kappa_s-\iota_s\right) x_s ds\right]
\end{align}
and the required conclusion now follows from \eqref{eq:integrability CT} by dominated convergence. Note that if $\kappa_t-\iota_t\geq 0$ then \eqref{eq:integrability CT} becomes superfluous since the result can then be obtained by monotone convergence. However, nothing guarantees that the perpetual futures price process in \eqref{eq:future price CT} is well defined under this weaker assumption.

To establish the second part of the statement note that, under the given additional assumption, we have
\begin{align}
	E^{Q_a}_t\left[e^{-\int_t^{\theta_t} \iota_u du}  x_{\theta_t}\right]
	&=E^{Q_a}_t\left[\int_t^\infty  e^{-\int_t^{s} \iota_u du}  x_{s}Q_{a}\left(\theta_t \in ds\middle| \F\right) \right]\\
	&=E^{Q_a}_t\left[\int_t^\infty  e^{-\int_t^{s} \kappa_u du}\left(\kappa_s-\iota_s\right)  x_{s}ds\right]=f_t
\end{align}
where the last equality follows from the first part of the statement. The last part directly follows by letting $\kappa$ be constant and setting $\iota=0$.
\end{proof}

\begin{proof}[Proof of Proposition \ref{proposition:constant coefficients CT}]
Under the stated assumption
\begin{align}
	 E&^{Q_a}\left[ \int_0^\infty e^{-\kappa s}\left|\kappa-\iota\right| x_sds\right]\\
	&= E^{Q_a}\left[ \int_0^\infty e^{-\left(\kappa_u-r_{a}+r_{b}\right)s}\left(\kappa-\iota\right) \left(\frac{B_{bs}}{B_{as}}x_s\right)ds\right]\\
	&= E^{Q_b}\left[\int_0^\infty e^{-\left(\kappa_u-r_{a}+r_{b}\right)s} \left(\kappa-\iota\right) x_0 ds\right]=\frac{\left(\kappa-\iota\right)x_0}{\kappa-r_a+r_b} <\infty
\end{align}
where the third equality follows $\kappa-r_{a}+r_{b}>0$ and the fact that pricing measures are related by \eqref{eq:relation Qa/Qb in CT}. This shows that condition \eqref{eq:integrability CT} is satisfied. Therefore, Theorem \ref{theorem:futures price CT-1} and the same change of probability now imply that 
\begin{align}
	f_t
	&=E^{Q_a}_{t}\left[ \int_t^\infty e^{-\kappa(s-t)}\left(\kappa-\iota\right)x_sds\right]\\
	&=E^{Q_a}_{t}\left[ \int_t^\infty e^{-\left(\kappa-r_a+r_b\right)(s-t)}\left(\kappa-\iota\right) \left(\frac{B_{bs}B_{at}}{B_{as}B_{bt}}x_s\right) ds\right]\\
	&= E^{Q_b}_{t}\left[ \int_t^\infty e^{-\left(\kappa-r_a+r_b\right)(s-t)}\left(\kappa-\iota\right) x_t ds\right]=\frac{\left(\kappa-\iota\right)x_t}{\kappa-r_a+r_b}. 
\end{align}
This establishes the required pricing formula and the comparative statics now follow by differentiating the result.
\end{proof}

\begin{proof}[Proof of Corollary \ref{corollary:equalizing CT}]
Under the stated assumption
\begin{align}
	 E^{Q_a}&\left[ \int_0^\infty e^{-\int_0^{s} \kappa_u du}\left|\kappa_s-\iota_s\right| x_sds\right]\\
	&= -E^{Q_a}\left[ \int_0^\infty   \left(\frac{B_{bs}}{B_{as}}x_s\right) d \left(e^{-\int_0^{s} \left(\kappa_u-r_{au}+r_{bu}\right) du}\right)\right]\\
	&= -E^{Q_b}\left[ \int_0^\infty  x_0 d \left(e^{-\int_0^{s} \left(\kappa_u-r_{au}+r_{bu}\right) du}\right) ds\right]=x_0<\infty
\end{align}
where the third equality follows from $\kappa_t-r_{at}+r_{bt}>0$ and the second from the fact that the pricing measures are related by
\begin{align}\label{eq:relation Qa/Qb in CT}
	\left.\frac{dQ_b}{dQ_a}\right|_{\F_s}=\frac{B_{bs}}{B_{as}}\frac{x_s}{x_0},\qquad 0\leq s<\infty.
\end{align}
This shows that condition \eqref{eq:integrability CT} is satisfied. Therefore, Theorem \ref{theorem:futures price CT-1} and the same change of probability now imply that 
\begin{align}
	f_t
	&=E^{Q_a}_{t}\left[ \int_t^\infty e^{-\int_t^{s} \kappa_u du}\left(\kappa_s-\iota_s\right)x_sds\right]\\
	&=-E^{Q_a}_{t}\left[ \int_t^\infty  \left(\frac{B_{bs}B_{at}}{B_{as}B_{bt}}x_s\right)   d \left(e^{-\int_t^{s} \left(\kappa_u-r_{au}+r_{bu}\right)du}\right)
\right]\\
	&=- E^{Q_b}_{t}\left[ x_t \int_t^\infty   d \left(e^{-\int_t^{s} \left(\kappa_u-r_{au}+r_{bu}\right)du}\right)\right]=x_t.
\end{align}
To establish the second part consider the self-financing strategy that starts from zero at some stopping time $\tau$, is long one unit of the contract, short $n_{t}=-1/B_{bt}$ units of the $b-$riskfree asset, and invests the remainder in the $a-$riskfree asset. Under the given specification the value of this strategy evolves according to
\begin{align}
  dv_t
  &=r_{at}\left(v_t-n_t B_{bt} x_t\right)dt + n_t d\left(B_{bt}x_t\right)+dx_t-\left(r_{at}-r_{bt}\right)x_t dt\\
  &=r_{at}\left(v_t+x_t\right)dt - \left(dx_t+ r_{bt}x_tdt\right)+dx_t-\left(r_{at}-r_{bt}\right)x_t dt=r_{at}v_t dt
\end{align}
subject to the initial condition $v_\tau=0$. This readily implies that $v_t=0$ for all $t\geq \tau$ and completes the proof.
\end{proof}

\begin{proof}[Proof of Proposition \ref{proposition:everlasting_call_put_prices}]
Equation \eqref{eq:everlasting option as expectation} and Lemma \ref{lemma:useful} below show that the everlasting call price $c_t=c(x_t)$ is the unique solution to
\begin{align}
  \kappa c'(x)=(r_a-r_b) x c'(x)+\frac12 x^2\|\sigma_x\|^2 c''(x)+\kappa\left(x-K\right)^+
\end{align}
in the set of linearly growing, $C^1$, and piecewise $C^2$ functions on $\R_+$. Standard results show that the general solution to this equation is 
\begin{align}
c(x)=\begin{cases}
x^\Theta    A_0  +x^\Pi A_1 , & x\leq  K,\\
x^\Theta    A_2  +x^\Pi A_3 +f(x)-K, & x > K,
  \end{cases}
\end{align}
for some constants $(A_i)_{i=0}^3$ where the exponents are defined as in the statement. Since the solution we seek grows at most linearly we must have that $A_1=A_2=0$. On the other hand, requiring that the solution is $C^1$ at the strike delivers a linear system for $(A_0,A_4)$ and solving that system gives the stated result.
\end{proof}

\noindent
Let $\rho>0$ and consider the function
\begin{align}\label{eq:definition v}
v(x_t):=E_t\left[\int_t^{\infty} e^{-\rho(s-t)} \ell\left(x_s\right) d s\right]
\end{align}
where $\ell:\R_+\to\R$ satisfies a \emph{linear growth condition} and $x_t$ is a geometric Brownian motion with parameters $(\mu, \sigma)$ such that $\rho-\mu>0$.

\begin{lemmaA}\label{lemma:useful}
The function $v(x)$ in \eqref{eq:definition v} is the unique solution to
\begin{align}\label{eq:ODE for w}
\rho w(x)=\mu x w^{\prime}(x)+\frac{1}{2} \sigma^2 x^2 w^{\prime \prime}(x)+\ell(x)
\end{align}
in the space of $C^1$ and piecewise $C^2$ functions on $\R_+$ that satisfy a linear growth condition.
\end{lemmaA}

\begin{proof}
The existence of a solution with the required properties follows from standard results on second order ODEs. Let $w(x)$ denote any such a solution. Since $w(x)$ solves \eqref{eq:ODE for w} it follows from It\^o's lemma that
\begin{align}
e^{-\rho t} w\left(x_t\right)+\int_0^t e^{-\rho s} \ell\left(x_s\right) d s
\end{align}
is a local martingale. Therefore, for any given $0\leq t<\infty$ there exist an increasing sequence of stopping times $(\tau_n)_{n=1}^\infty$ such that $t\leq \tau_n \to \infty$ and
\begin{align}\label{eq:w martingale}
w\left(x_t\right)=E_t\left[\int_t^{\tau_n} e^{-\rho s} \ell\left(x_s\right) d s+e^{-\rho \tau_n} w\left(x_{\tau_n}\right)\right].
\end{align}
Since $w$ grows at most linearly and $\rho>\mu^+$ we have that
\begin{align}
e^{-\rho \tau_n}\left|w\left(x_{\tau_n}\right)\right| \leq  \alpha y_{\tau_n}
\end{align}
for some $\alpha>0$ where the process $y_t=e^{-\rho t}\left(1+x_t\right)$ is supermartingale of class $\mathbb{D}$ that converges to zero. This immediately implies that
\begin{equation}
\lim _{n \rightarrow \infty}\left|E_t\left[e^{-\rho \tau_n} w\left(x_{\tau_n}\right)\right]\right| \leq \lim _{n \rightarrow \infty} E_t\left[y_{\tau_n}\right]=E_t\left[\lim _{n \rightarrow \infty} y_{\tau_n}\right]=0. \label{eq:v_upper_bound}
\end{equation}
On other hand, since $\ell$ grows at most linearly we have that
\begin{align}
E_t\left[\int_t^{\tau_n} e^{-\rho s} \left|\ell\left(x_s\right)\right| d s\right]&\leq 
E_t\left[\int_t^{\infty} e^{-\rho s}\left|\ell\left(x_s\right)\right| d s\right] \\
&\leq \alpha E_t\left[\int_t^{\infty} e^{-\rho s}\left(1+x_s\right) d s\right]=\frac{\alpha x_t}{\rho-\mu}+\frac{\alpha}{\rho}
\end{align}
for some $\alpha>0$ and, therefore, 
\begin{align}
\lim _{n \rightarrow \infty} E_t\left[\int_t^{\tau_n} e^{-\rho s} \ell\left(x_s\right) d s\right]
=E_t\left[\int_t^{\infty} e^{-\rho s} \ell\left(x_s\right) d s\right]
\end{align}
by dominated convergence. Using these properties to pass to limit in \eqref{eq:w martingale} then shows that $w=v$ which also establishes the uniqueness claim.
\end{proof}

% subsection continuous_time (end)

% section proofs (end)

\section{An incorrect cash flow specification} % (fold)
\label{sec:An erroneous cash flow specification}
In the original version of their paper \citeauthor*{he2022fundamentals} show that in a continuous-time model with constant coefficients and $r_b=\iota= 0$ the perpetual futures price is  $\left(1+\frac{r_a}{\kappa}\right)x_t$ which is clearly different from the price $x_t/\left(1-\frac{r_a}{\kappa}\right) $
prescribed by Proposition \ref{proposition:constant coefficients CT} in this case. As we show below, the reason for this discrepancy is that the cash flow specification of \citeauthor{he2022fundamentals} is different from ours and, in fact, inconsistent with the assumption that entering a contract is costless (unless $r_a= 0$). 

\citet{he2022fundamentals-original} assume that the cumulative \emph{discounted} cash flow generated by a long position in one perpetual futures contract between two consecutive stopping times $\tau$ and $\sigma\geq \tau$ is 
\begin{align}\label{eq:incorrect spec}
e^{-r_a \sigma}\left(F_{\sigma}-F_\tau\right)-\int_\tau^\sigma{e^{-r_a s} \kappa \left(F_s-x_s\right) ds }.
\end{align}
where $F_t$ is the perpetual futures price. If entering the futures contract is costless then the value of this payoff at date $\tau$ should be zero for all stopping times $\tau\leq \sigma$, i.e. we should have 
\begin{align}\label{eq:restriction He et al.}
  E^{Q_a}_\tau \left[e^{-r_a \sigma}\left(F_{\sigma}-F_\tau\right)-\int_\tau^\sigma{e^{-r_a s} \kappa \left(F_s-x_s\right) ds }\right]=0, \qquad \tau\leq \sigma\in\S.
\end{align}
But this restriction implies that 
\begin{align}
  M_t(T)=F_{t\wedge T}-\int_{0}^{t\wedge T} e^{r_a (T-s)}\kappa \left(F_s-x_s\right)ds
\end{align}
is a martingale under $Q_a$ for any $T<\infty$ which is not possible if $r_a\neq 0$. Indeed, if that property was satisfied then the difference
\begin{align}
  M_t(T)-M_t(T+h)=\int_{0}^{t} \left(e^{r_a (T-s)}-e^{r_a (T+h-s)}\right)\kappa \left(F_s-x_s\right)ds 
\end{align}
would be a continuous $Q_a-$martingale of finite variation on $[0,T]$ and thus a constant. If $r_a=0$ then this is not a problem. However, if $r_a\neq 0$ then the constancy of the above difference requires that $F_t=x_t$ which is inconsistent with \eqref{eq:restriction He et al.} because the exchange rate \emph{cannot } be a martingale under $Q_a$ when $r_a\neq 0=r_b$. Intuitively, the problem with the specification in \eqref{eq:incorrect spec} is that, instead of being paid continuously, the futures margin $F_{\sigma}-F_\tau$ is paid in a lumpsum upon exiting the contract. 

% section section_name (end)

% bibligraphy

\setlength{\bibsep}{0cm} 
\bibliographystyle{plainnat} 
\small 
\bibliography{AHJ-bib.bib}

\end{document}

%% file: AHJ-Figure1.tex
 
  
  \begin{tikzpicture}
    \begin{groupplot}[
      mlineplot,
      scale only axis,
      axis on top,
      width=0.75\textwidth,
      height=0.13\textheight,	  
      group style={group size= 1 by 3}]
    
  \nextgroupplot[
      mlineplot,
      scale only axis,
      ymax=100,
      ymin=0,
      ytick={0,50,...,100},
      date coordinates in=x,
      ylabel={Volume $(\$\mathrm{b})$},
      table/col sep=comma,
      clip=true,
      ybar,
      bar width=0.02cm,
      xtick style={opacity=0},
      x axis line style={opacity=0},
      date ZERO=2019-01-01,
      xmin=2019-01-01,
      xmax=2024-01-01,
      xtick={
      2019-01-01,
      2020-01-01,
      2021-01-01,
      2022-01-01,
      2023-01-01,
      2024-01-01},
      extra x ticks={
      2019-06-30,
      2020-06-30,
      2021-06-30,
      2022-06-30,
      2023-06-30,
      2024-06-30},
      extra x tick style={	  		  
              grid=none,
              anchor=north},
      extra x tick labels={2019,2020,2021,2022,2023},
      xticklabels={},
      legend pos=north west,
	  axis on top=false]

    \addplot[area legend,draw=aqua2,ultra thin,fill=aqua2]  
    table[x index=0,y expr=\thisrowno{1}/1000000000]{data/perp_volume_margin.csv};

    \legend{Inverse}

    \nextgroupplot[
      ymax=100,
      ymin=0,
      ytick={0,50,...,100},
      date coordinates in=x,
      ylabel={Volume $(\$\mathrm{b})$},
      table/col sep=comma,
      date ZERO=2019-01-01,
      xmin=2019-01-01,
      xmax=2024-01-01,
      xtick={
      2019-01-01,
      2020-01-01,
      2021-01-01,
      2022-01-01,
      2023-01-01,
      2024-01-01},
      extra x ticks={
      2019-06-30,
      2020-06-30,
      2021-06-30,
      2022-06-30,
      2023-06-30,
      2024-06-30},
      xticklabels={},
      clip=true,
      ybar,
      bar width=0.02cm,
      xtick style={opacity=0},
      x axis line style={opacity=0},
      extra x tick style={
        xtick style={opacity=0},
        grid=none,anchor=west},
      extra x tick labels={2019,2020,2021,2022,2023},      
      legend pos=north west,
	  axis on top=false]

    \addplot[area legend,draw=aqua1,fill=aqua1,ultra thin] 
    table[x index=0,y expr=(\thisrowno{2}+0.01)/1000000000]{data/perp_volume_margin.csv};
    
    \legend{Linear}
    
		\nextgroupplot[
			xmin=A,
			xmax=E,
			ymin=0,
			ymax=50,
			table/col sep=comma, 
			bar width=0.45cm,
			xmajorgrids=true, 
			xticklabel style={font=\scriptsize},
			x axis line style={opacity=0},
			ytick={0,25,50},
			yticklabels={0,25,50},
			symbolic x coords={A,dYdX,Kraken,Gate.io,Huobi,Phemex,OKX,Deribit,Bybit,Binance,BitMEX,E},
			xtick={A,dYdX,Kraken,Gate.io,Huobi,Phemex,OKX,Deribit,Bybit,Binance,BitMEX,E},
			xticklabels={,dYdX,Kraken,Gate.io,Huobi,Phemex,OKX,Deribit,Bybit,Binance,BitMEX,},
			ybar=0pt,
			xtick style={opacity=0},
			legend pos=north west,
      	  	ylabel={Share $(\%)$},
	  	  	axis on top=false,
			clip=false]

    \addplot[area legend,fill=aqua2,draw=aqua2] table[x index=0,y expr=100*\thisrowno{1}]{data/market-shares.csv};
    \addplot[area legend,fill=aqua1,draw=aqua1] table[x index=0,y expr=\thisrowno{2}*100]{data/market-shares.csv};
	\legend{Volume,Open interest}
		
    \end{groupplot}
    
		\end{tikzpicture}

%% file: AHJ-Figure2.tex
    \begin{tikzpicture}
		
         \begin{axis}[
           mlineplot,
	       width=0.6\textwidth,
       	   height=0.6\textwidth,
   	       /pgf/declare function={LF(\k,\ra,\rb)=(\k*(1+\rb))/(\rb-\ra+\k*(1+\rb));},
           ylabel={Futures/spot price: $f_t/x_t$},
		   xlabel={Funding premium: $\kappa$},
		   ymax=1.0005,
		   ymin=0.9995,
		   xmin=0.2,
		   xmax=1,
		   xminorgrids=true,
		   yminorgrids=true,
		   minor tick num=1,
		   clip=true,		   
		   cycle list={
		       {very thick,aqua1,mark=*,mark options={mark size=1.9,thick}},
		       {very thick,aqua2,mark=square*,mark options={mark size=1.9,thick}},
		       {very thick,black,mark=none,mark options={mark size=1.9,thick}},
		       {very thick,aqua2,mark=square*,mark options={mark size=1.9,thick,fill=white}},
			   {very thick,aqua1,mark=*,mark options={mark size=1.9,thick,fill=white}}
		     },
		   legend pos=north east,
		   mark repeat=25,
		   xtick={0.0,0.2,0.4,...,1},
		   ytick={0.9995,1,1.0005},
		   y tick label style={
		   	/pgf/number format/fixed,
		   	/pgf/number format/fixed zerofill, 
			/pgf/number format/precision=4}
		   ]
		   
		   \addplot table[x index=0,y index=1] {data/Essai.csv};\addlegendentry{$(r_a,r_b)=(0.10;0)\Delta$};
		   \addplot table[x index=0,y index=2] {data/Essai.csv};\addlegendentry{$(r_a,r_b)=(0.05;0)\Delta$};
		   \addplot table[x index=0,y index=3] {data/Essai.csv};\addlegendentry{$(r_a,r_b)=(0;0)$};
		   \addplot table[x index=0,y index=4] {data/Essai.csv};\addlegendentry{$(r_a,r_b)=(0,0.05)\Delta$};
		   \addplot table[x index=0,y index=5] {data/Essai.csv};\addlegendentry{$(r_a,r_b)=(0,0.10)\Delta$};

		  % \addplot {LF(\x,0.1/(3*360),0.00)} ;\addlegendentry{$(r_a,r_b)=(0.10;0)\Delta$}
		  % \addplot {LF(\x,0.05/(3*360),0.00)};\addlegendentry{$(r_a,r_b)=(0.05;0)\Delta$}
		  % \addplot {LF(\x,0.00,0.00)}        ;\addlegendentry{$(r_a,r_b)=(0;0)$};
		  % \addplot {LF(\x,0.00,0.05/(3*360))};\addlegendentry{$(r_a,r_b)=(0,0.05)\Delta$};
		  % \addplot {LF(\x,0.00,0.1/(3*360))} ;\addlegendentry{$(r_a,r_b)=(0,0.10)\Delta$};

      \end{axis}
  \end{tikzpicture}
\caption{Perpetual futures price}  
\fignotes{This figure plots the ratio $f_t/x_t$ of the perpetual (linear) futures price to the spot price as a function of the funding premium $\kappa$ for a contract with interest factor $\iota=0$ and different interest rate configurations
$\{r_a,r_b\in \{0,0.05,0.1\}\Delta:r_ar_b=0\}$ where the multiplication by $\Delta=\frac1{3\cdot 360}$ converts annual rates into 8h period rates. 
% The figure starts from the funding premium $\kappa =0.01$ to ensure that the integrability condition \eqref{eq:integrability DT-constant} holds for all values of the interest rates.
}

%% file: AHJ-Figure3.tex
    \begin{tikzpicture}
         \begin{axis}[
       mlineplot,
       width=0.6\textwidth,
   	   height=0.6\textwidth,
       /pgf/declare function={LF(\k,\ra,\rb)=(\k*(1+\rb))/(\rb-\ra+\k*(1+\rb));},
       ylabel={Futures/spot price: $f_t/x_t$},
	   xlabel={Funding premium: $\kappa$},
	   ymax=1.0005,
	   ymin=0.9995,
	   xmin=0.2,
	   xmax=1,
	   xminorgrids=true,
	   yminorgrids=true,
	   minor tick num=1,
	   clip=true,		   
	   cycle list={
	       {very thick,aqua1,mark=*,mark options={mark size=1.9,thick}},
	       {very thick,aqua2,mark=square*,mark options={mark size=1.9,thick}},
	       {very thick,black,mark=none,mark options={mark size=1.9,thick}},
	       {very thick,aqua2,mark=square*,mark options={mark size=1.9,thick,fill=white}},
		   {very thick,aqua1,mark=*,mark options={mark size=1.9,thick,fill=white}}
	     },
	   legend pos=north east,
	   mark repeat=25,
	   xtick={0.0,0.2,0.4,...,1},
	   ytick={0.9995,1,1.0005},
	   y tick label style={
	   	/pgf/number format/fixed,
	   	/pgf/number format/fixed zerofill, 
		/pgf/number format/precision=4}
		   ]
		   \addplot table[x index=0,y index=6] {data/Essai.csv};\addlegendentry{$(r_a,r_b)=(0.10;0)\Delta$};
		   \addplot table[x index=0,y index=7] {data/Essai.csv};\addlegendentry{$(r_a,r_b)=(0.05;0)\Delta$};
		   \addplot table[x index=0,y index=8] {data/Essai.csv};\addlegendentry{$(r_a,r_b)=(0;0)$};
		   \addplot table[x index=0,y index=9] {data/Essai.csv};\addlegendentry{$(r_a,r_b)=(0,0.05)\Delta$};
		   \addplot table[x index=0,y index=10] {data/Essai.csv};\addlegendentry{$(r_a,r_b)=(0,0.10)\Delta$};
		   
		  % \addplot {IF(\x/(3*360),0.1/(3*360),0.00)};\addlegendentry{$(r_a,r_b)=(0.1;0)h$};
		  % \addplot {IF(\x/(3*360),0.05/(3*360),0.00)};\addlegendentry{$(r_a,r_b)=(0.05;0)h$};
		  % \addplot {IF(\x/(3*360),0.00,0.00)};\addlegendentry{$(r_a,r_b)=(0;0)$};
		  % \addplot {IF(\x/(3*360),0.00,0.05/(3*360))};\addlegendentry{$(r_a,r_b)=(0,0.05)h$};
		  % \addplot {IF(\x/(3*360),0.00,0.1/(3*360))};\addlegendentry{$(r_a,r_b)=(0,0.1)h$};
		 
      \end{axis}
  \end{tikzpicture}
\caption{Perpetual inverse futures price}  
\fignotes{This figure plots the ratio $f_{It}/x_t$ of the perpetual inverse futures price to the spot price as a function of the funding premium $\kappa$ for a contract with interest factor $\iota=0$ and different interest rate configurations
$\{r_a,r_b\in \{0,0.05,0.1\}\Delta:r_ar_b=0\}$ where the multiplication by $\Delta=\frac1{3\cdot 360}$converts annual rates into 8h period rates. 
% The figure starts from the funding premium $\kappa =0.01$ to ensure that the integrability condition \eqref{eq:integrability DT-constant inverse} holds for all values of the interest rates.
}

%% file: AHJ-Figure4.tex
    \begin{tikzpicture}
    
    \begin{groupplot}[
      mlineplot,
      xmin=0,
      xmax=2,
      ymin=0,
      ymax=1,
      clip,
      axis on top=false,
      width=0.425\textwidth,
      height=0.55\textwidth,	  
      group style={group size= 2 by 1,horizontal sep=0.11\textwidth},
      xlabel={Underlying: $x$},
      cycle list={
            {very thick,aqua2,densely dashed},
            {ultra thick,aqua1}},
       legend style={draw=none,fill=none,font=\footnotesize},
       xtick={0,1,2},
       xticklabels={0,$K=1$,2},
       ytick={0,0.5,1},
       xminorgrids=true,
       yminorgrids=true,
       minor tick num=1,
       ]
       
       \nextgroupplot[
       ylabel={Call price},
       legend pos=north west]
       \addplot table[x index=0,y expr=\thisrowno{2}/100] {data/EVCalls.dat};      
       \addplot table[x index=0,y expr=\thisrowno{1}/100] {data/EVCalls.dat};
       \addplot table[x index=0,y expr=\thisrowno{2}/100] {data/EVCalls.dat};  
       \addplot[domain=0:2,thick] {max(0,x-1)} node[pos=0.7,below,sloped,font=\scriptsize] {Payoff: $(x-1)^+$};

       \legend{European,Everlasting}

       \nextgroupplot[
       ylabel={\footnotesize{Call Delta}},
       legend pos=south east]
       \addplot table[x index=0,y expr=\thisrowno{5}/100] {data/EVCalls.dat};      
       \addplot table[x index=0,y expr=\thisrowno{4}/100] {data/EVCalls.dat};  

       % \legend{European,Everlasting}

    \end{groupplot}

  \end{tikzpicture}
\caption{Everlasting call options}  
\fignotes{This figure compares the arbitrage prices (left) and deltas (right) of an everlasting call option with multiplier $\kappa_o=1$ and a European call option with maturity $T=1/\kappa_o$ in the Black-Scholes-Merton model with $r_a=r_b\equiv 0$ and $\|\sigma_x\|=1$. The two options have the same strike equal to one.}\label{fig:everlasting call}

%% file: AHJ-Perpetual-Futures-Pricing-R1.bbl
\begin{thebibliography}{29}
\providecommand{\natexlab}[1]{#1}
\providecommand{\url}[1]{\texttt{#1}}
\expandafter\ifx\csname urlstyle\endcsname\relax
  \providecommand{\doi}[1]{doi: #1}\else
  \providecommand{\doi}{doi: \begingroup \urlstyle{rm}\Url}\fi

\bibitem[Alexander et~al.(2020)Alexander, Choi, Park, and Sohn]{alexander2020bitmex}
Carol Alexander, Jaehyuk Choi, Heungju Park, and Sungbin Sohn.
\newblock Bitmex bitcoin derivatives: Price discovery, informational efficiency, and hedging effectiveness.
\newblock \emph{Journal of Futures Markets}, 40\penalty0 (1):\penalty0 23--43, 2020.

\bibitem[Alexander et~al.(2023)Alexander, Chen, and Imeraj]{alexander2023}
Carol Alexander, Ding Chen, and Arben Imeraj.
\newblock Crypto quanto and inverse options.
\newblock \emph{Mathematical Finance}, 33\penalty0 (4):\penalty0 1005--1043, 2023.

\bibitem[Angeris et~al.(2023)Angeris, Chitra, Evans, and Lorig]{angeris2023primer}
Guillermo Angeris, Tarun Chitra, Alex Evans, and Matthew Lorig.
\newblock A primer on perpetuals.
\newblock \emph{SIAM Journal on Financial Mathematics}, 14\penalty0 (1):\penalty0 17--30, 2023.

\bibitem[Bankman-Fried and White(2021)]{paradigm2021everlasting}
Sam Bankman-Fried and Dave White.
\newblock Everlasting options, May 2021.
\newblock URL \url{https://www.paradigm.xyz/2021/05/everlasting-options}.
\newblock [Blog post].

\bibitem[Borodin and Salminen(1996)]{borodin}
Andrei Borodin and Paavo Salminen.
\newblock \emph{Handbook of Brownian Motion: Facts and Formulae}.
\newblock Birkh{\"a}user Verlag, 1996.

\bibitem[Christin et~al.(2023)Christin, Routledge, Soska, and Zetlin-Jones]{Nicholas023}
Nicholas Christin, Bryan Routledge, Kyle Soska, and Ariel Zetlin-Jones.
\newblock The crypto carry trade.
\newblock Preprint at \url{http://gerbil.life/papers/CarryTrade.v1.2.pdf}, 2023.

\bibitem[Clark(2023)]{Clark2}
Joseph Clark.
\newblock Spanning with power perpetuals.
\newblock \emph{Opyn research paper}, 2023.
\newblock URL \url{https://ssrn.com/abstract=4317072}.

\bibitem[Clark et~al.(2024)Clark, Leone, and Robinson]{Clark}
Joseph Clark, Andrew Leone, and Dan Robinson.
\newblock Everything is a {P}erp, March 2024.
\newblock URL \url{https://research.opyn.co/everything-is-a-perp}.
\newblock [Blog post].

\bibitem[CME(2024)]{CME}
CME.
\newblock Overview of crypto-currency products, August 2024.
\newblock URL \url{https://www.cmegroup.com/markets/cryptocurrencies.html\#explore-our-cryptocurrency-products}.

\bibitem[De~Blasis and Webb(2022)]{deblasis2022arbitrage}
Riccardo De~Blasis and Alexander Webb.
\newblock Arbitrage, contract design, and market structure in bitcoin futures markets.
\newblock \emph{Journal of Futures Markets}, 42\penalty0 (3):\penalty0 492--524, 2022.

\bibitem[Gehr~Jr.(1988)]{gehr1988undated}
Adam Gehr~Jr.
\newblock Undated futures markets.
\newblock \emph{The Journal of Futures Markets}, 8\penalty0 (1):\penalty0 89, 1988.

\bibitem[Hayes(2016)]{BitMEX-Intro}
Arthur Hayes.
\newblock Why is {XBTUSD} is a superior trading product?, May 2016.
\newblock URL \url{https://blog.bitmex.com/why-xbtusd-is-a-superior-trading-product}.
\newblock [Blog post].

\bibitem[Hayes(2018)]{BitMEX-Q}
Arthur Hayes.
\newblock Why quanto?, August 2018.
\newblock URL \url{https://blog.bitmex.com/why-quanto}.
\newblock [Blog post].

\bibitem[Hayes(2021)]{Naka}
Arthur Hayes.
\newblock All aboard!, April 2021.
\newblock URL \url{https://blog.bitmex.com/all-aboard}.
\newblock [Blog post].

\bibitem[He et~al.(2022)He, Manela, Ross, and von Wachter]{he2022fundamentals-original}
Songrun He, Asaf Manela, Omri Ross, and Victor von Wachter.
\newblock Fundamentals of perpetual futures, 2022.
\newblock [First version].

\bibitem[He et~al.(2024)He, Manela, Ross, and von Wachter]{he2022fundamentals}
Songrun He, Asaf Manela, Omri Ross, and Victor von Wachter.
\newblock Fundamentals of perpetual futures.
\newblock Preprint at \url{https://papers.ssrn.com/sol3/papers.cfm?abstract_id=4301150}, August 2024.

\bibitem[Jeanblanc et~al.(2009)Jeanblanc, Yor, and Chesney]{jeanblanc}
Monique Jeanblanc, Marc Yor, and Marc Chesney.
\newblock \emph{Mathematical Methods for Financial Markets}.
\newblock Springer Finance. Springer London, 2009.

\bibitem[Johnson and Zervos(2007)]{Zervos}
Timothy Johnson and Mihail Zervos.
\newblock The solution to a second order linear ordinary differential equation with a non-homogeneous term that is a measure.
\newblock \emph{Stochastics: An International Journal of Probability and Stochastic Processes}, 79:\penalty0 363--382, 2007.

\bibitem[Madrigal-Cianci and Kristensen(2022)]{Kristensen}
Juan Madrigal-Cianci and Jesper Kristensen.
\newblock Time-efficient decentralized exchange of everlasting options with exotic payoff functions.
\newblock In \emph{2022 IEEE International Conference on Blockchain (Blockchain)}, pages 427--434, 2022.

\bibitem[Notte(2022)]{Medium-3}
Crypto Notte.
\newblock Making money with delta-neutral trading using perpetual swaps, May 2022.
\newblock URL \url{https://crypto-notte.medium.com/making-money-out-of-delta-neutral-trading-with-perpetual-swaps-9ce2759fc6f}.
\newblock [Blog post].

\bibitem[Palepu(2020{\natexlab{a}})]{Medium-1}
Aditya Palepu.
\newblock What are perpetual swaps?, May 2020{\natexlab{a}}.
\newblock URL \url{https://medium.com/derivadex/what-are-perpetual-swaps-130236587df2}.
\newblock [Blog post].

\bibitem[Palepu(2020{\natexlab{b}})]{Medium-2}
Aditya Palepu.
\newblock What is the funding rate for perpetual swaps?, June 2020{\natexlab{b}}.
\newblock URL \url{https://medium.com/derivadex/what-is-the-funding-rate-for-perpetual-swaps-a0335c4228a9}.
\newblock [Blog post].

\bibitem[{Perennial Labs}(2024)]{Perennial}
{Perennial Labs}.
\newblock Power perpetuals on perennial, June 2024.
\newblock URL \url{https://medium.com/perennial-protocol/power-perpetuals-on-perennial-606ec8cdd422}.
\newblock [Blog post].

\bibitem[Prospere(2022)]{medium2022sqeeth}
Wade Prospere.
\newblock Squeeth primer: {A} guide to understanding {O}pyn’s implementation of {S}queeth, January 2022.
\newblock URL \url{https://medium.com/opyn/squeeth-primer-a-guide-to-understanding-opyns-implementation-of-squeeth-a0f5e8b95684}.
\newblock [Blog post].

\bibitem[Schmeling et~al.(2022)Schmeling, Schrimpf, and Todorov]{schmeling2022crypto}
Maik Schmeling, Andreas Schrimpf, and Karamfil Todorov.
\newblock Crypto carry.
\newblock Preprint at \url{https://ssrn.com/abstract=4268371}, 2022.

\bibitem[Shiller(1993)]{shiller1993measuring}
Robert Shiller.
\newblock Measuring asset values for cash settlement in derivative markets: {H}edonic repeated measures indices and perpetual futures.
\newblock \emph{The Journal of Finance}, 48\penalty0 (3):\penalty0 911--931, 1993.

\bibitem[Szpruch et~al.(2024)Szpruch, Xu, Sabate-Vidales, and Aouane]{java}
Lukasz Szpruch, Jiahua Xu, Marc Sabate-Vidales, and Kamel Aouane.
\newblock Leveraged trading via lending platforms.
\newblock \emph{University of Edinburgh working paper}, 2024.
\newblock URL \url{https://ssrn.com/abstract=4713126}.

\bibitem[White(2021)]{White-Cartoon}
Dave White.
\newblock The cartoon guide to perps, March 2021.
\newblock URL \url{https://www.paradigm.xyz/2021/03/the-cartoon-guide-to-perps}.
\newblock [Blog post].

\bibitem[White et~al.(2021)White, Robinson, Koticha, Leone, Gauba, and Krishnan]{paradigm2021power}
Dave White, Dan Robinson, Zubin Koticha, Andrew Leone, Alexis Gauba, and Aparna Krishnan.
\newblock Power perpetuals, August 2021.
\newblock URL \url{https://www.paradigm.xyz/2021/08/power-perpetuals}.
\newblock [Blog post].

\end{thebibliography}
